\documentclass[twocolumn]{IEEEtran}

\usepackage{amsmath,epsfig,amssymb,verbatim,amsopn,cite,multirow}
\usepackage{amsthm}
\usepackage{balance}
\usepackage{multirow}

\usepackage[usenames,dvipsnames]{color}
\usepackage[all]{xy}  
\usepackage{url}
\usepackage{amsfonts}
\usepackage{amssymb}
\usepackage{epsfig}
\usepackage{epstopdf}
\usepackage{bm}
\usepackage{balance}
\usepackage{graphicx}
\usepackage{subcaption}
\usepackage{footnote}
\usepackage{algorithm,algorithmic}
\usepackage[margin=16mm,top=15mm,bottom=15mm]{geometry}

\usepackage[nodisplayskipstretch]{setspace}
\newtheorem{Lemma}{Lemma}
\newtheorem{Theorem}{Theorem}

\newtheorem{Definition}[Lemma]{Definition}

\newtheorem{Remark}{Remark}

  {\proof}{\proofend}
\newtheorem{proposition}{Proposition}

\newcommand{\qa}{{\bf a}}
\newcommand{\qb}{{\bf b}}

\newcommand{\qe}{{\bf e}}

\newcommand{\qg}{{\bf g}}

\newcommand{\qr}{{\bf r}}

\newcommand{\qv}{{\bf v}}
\newcommand{\qw}{{\bf w}}
\newcommand{\qx}{{\bf x}}
\newcommand{\qy}{{\bf y}}

\newcommand{\qF}{{\bf F}}
\newcommand{\qG}{{\bf G}}

\newcommand{\qI}{{\bf I}}

\newcommand{\qZ}{{\bf Z}}

\newcommand{\Sm}{\mathcal{S}_m^{\dl}}
\newcommand{\Wm}{\mathcal{W}_m^{\dl}}
\newcommand{\Smu}{\mathcal{S}_m^{\ul}}
\newcommand{\Wmu}{\mathcal{W}_m^{\ul}}
\newcommand{\Zk}{\mathcal{Z}_k^{\dl}}
\newcommand{\Tk}{\mathcal{T}_k^{\dl}}

\newcommand{\Zqdl}{\mathcal{Z}_q^{\dl}}
\newcommand{\Tqdl}{\mathcal{T}_q^{\dl}}

\newcommand{\Zlup}{\mathcal{Z}_{\ell}^{\ul}}
\newcommand{\Tlup}{\mathcal{T}_{\ell}^{\ul}}

\newcommand{\Sn}{\sigma_n^2}
\newcommand{\MR}{\mathsf{MR}}
\newcommand{\ZF}{\mathsf{ZF}}
\newcommand{\FZF}{\mathsf{FZF}}
\newcommand{\PZF}{\mathsf{PZF}}

\newcommand{\vFD}{\mathsf{NAFD}}

\newcommand{\Ntx}{N}
\newcommand{\Nrx}{N}

\newcommand{\dl}{\mathtt{dl}}
\newcommand{\ul}{\mathtt{ul}}

\newcommand{\hgmkpd}{\hat{\qg}_{mk'}^{\dl}}
\newcommand{\tgmkd}{\tilde{\qg}_{mk}^{\dl}}
\newcommand{\tgmlu}{\tilde{\qg}_{m\ell}^{\ul}}

\newcommand{\wmkdl}{\qv_{mk}^{\dl}}
\newcommand{\wmkdlzf}{\qv_{mk}^{\dl,\ZF}}
\newcommand{\wmkdlmr}{\qv_{mk}^{\dl,\MR}}

\newcommand{\wmlulzf}{\qv_{m\ell}^{\ul,\ZF}}
\newcommand{\wmlulmr}{\qv_{m\ell}^{\ul,\MR}}

\newcommand{\wikdl}{\qv_{ik}^{\dl}}
\newcommand{\wmkpdl}{\qv_{mk'}^{\dl}}
\newcommand{\wmkpdlzf}{\qv_{mk'}^{\dl,\ZF}}
\newcommand{\wmkpdlmr}{\qv_{mk'}^{\dl,\MR}}
\newcommand{\wmlul}{\qv_{m\ell}^{\ul}}

\newcommand{\wiqdlzf}{\qv_{iq}^{\dl,\ZF}}
\newcommand{\wiqdlmr}{\qv_{iq}^{\dl,\MR}}

\newcommand{\dm}{\delta_{m}}

\newcommand{\gmqu}{{\qg}_{mq}^{\ul}}

\newcommand{\gmkd}{\qg_{mk}^{\dl}}

\newcommand{\hgmkd}{\hat{\qg}_{mk}^{\dl}}

\newcommand{\hgmonu}{\hat{\qg}_{m1}^{\ul}}
\newcommand{\hgmKuu}{\hat{\qg}_{mK_u}^{\ul}}

\newcommand{\hgmlu}{\hat{\qg}_{m\ell}^{\ul}}

\newcommand{\gmlu}{\qg_{m\ell}^{\ul}}

\newcommand{\gamdmk}{\gamma_{mk}^{\dl}}
\newcommand{\gamdmkp}{\gamma_{mk'}^{\dl}}
\newcommand{\gamuml}{\gamma_{m\ell}^{\ul}}
\newcommand{\gamumlp}{\gamma_{m\ell'}^{\ul}}

\newcommand{\hGmdl}{\hat{\qG}_{\mathcal{S}_m}^{\dl}}

\newcommand{\hGmul}{\hat{\qG}_{m}^{\ul}}

\newcommand{\betamkd}{\beta_{mk}^{\dl}}
\newcommand{\betamkpd}{\beta_{mk'}^{\dl}}

\newcommand{\betakldu}{\beta_{k\ell}^{\mathtt{du}}}

\newcommand{\betamlu}{\beta_{m\ell}^{\ul}}

\newcommand{\alphml}{\alpha_{m\ell}}

\newcommand{\SIm}{\sigma^2_{\mathtt{SI},m}}

\newcommand{\SINR}{\mathtt{SINR}}
\newcommand{\Ex}{\mathbb{E}}

\newcommand{\DSlpzf}{\mathbb{DS}_{\ul,\ell}^\PZF}
\newcommand{\UIqlpzf}{{\mathbb{UI}}_{\ul,q \ell} ^\PZF}
\newcommand{\MIqlpzf}{\mathbb{MI}_{\ul, q \ell}^{\PZF}}
\newcommand{\ANlpzf}{{\mathbb{AN}}_{\ul, \ell} ^\PZF}

\DeclareMathOperator{\F}{\mathbf{F}} 

\DeclareMathOperator{\aaa}{\mathbf{a}}

\DeclareMathOperator{\K}{\mathcal{K}}

\DeclareMathOperator{\SSS}{\mathcal{S}}
\DeclareMathOperator{\CN}{\mathcal{CN}}

\DeclareMathOperator{\bb}{\mathbf{b}}

\DeclareMathOperator{\MM}{\mathcal{M}}

\DeclareMathOperator{\DElta}{\boldsymbol{\delta}}

\DeclareMathOperator{\x}{\mathbf{x}}

\DeclareMathOperator{\VARSIGMA}{\boldsymbol{\varsigma}}

\DeclareMathOperator{\THeta}{\boldsymbol{\theta}}

\DeclareMathOperator{\ALPHA}{\boldsymbol{\alpha}}

\DeclareMathOperator{\UEdk}{\mathrm{UE}^{\mathtt{dl}}_{\mathit{k}}}
\DeclareMathOperator{\UEdkp}{\mathrm{UE}_{k'}^{\mathtt{dl}}}
\DeclareMathOperator{\UEul}{\mathrm{UE}_{\ell}^{\mathtt{ul}}}

\newtheorem{remark}{Remark}

\title{\huge Network-Assisted Full-Duplex Cell-Free Massive MIMO Systems Under Infeasible Circumstances}

 \author{Trinh Van Chien,~\IEEEmembership{Member,~IEEE}, Bui~Trong~Duc, Mohammadali~Mohammadi,~\IEEEmembership{Senior Member,~IEEE}, Hien~Quoc~Ngo,~\IEEEmembership{Fellow,~IEEE}, and Michail Matthaiou,~\IEEEmembership{Fellow, IEEE }
 \vspace{-0.5cm}
\thanks{
This research is funded by the Vietnam
National Foundation for Science and Technology Development (NAFOSTED)
under grant number 102.02-2025.73 for Trinh Van Chien. This work was supported by the U.K. Engineering and Physical Sciences Research Council (EPSRC) grant (EP/X04047X/2) for TITAN Telecoms Hub. The work of H.~Q.~Ngo was supported by the U.K. Research and Innovation Future Leaders Fellowships under Grant MR/X010635/1, and a research grant from the Department for the Economy Northern Ireland under the US-Ireland R\&D Partnership Programme. The work of M. Matthaiou was supported by the European Research Council (ERC) under the European Union’s Horizon 2020 Research and Innovation Programme (grant agreement No. 101001331)
\textit{(Corresponding authors: Bui Trong Duc; Mohammadali Mohammadi}).

Trinh Van Chien and Bui Trong Duc are with the School of Information and Communications Technology (SoICT), Hanoi University of Science and Technology (HUST), Vietnam (email:chientv@soict.hust.edu.vn, buitrongduc0502tlhpvn@gmail.com). 
 
 Mohammadali Mohammadi, Hien Quoc Ngo, and Michail Matthaiou are with the Centre for Wireless Innovation (CWI), Queen's University Belfast, U.K.
 (email:\{m.mohammadi, hien.ngo, m.matthaiou\}@qub.ac.uk). 
 
 Parts of this paper appeared at the 2024  GECCO conference~\cite{chien2025differential}.  
}}

\allowdisplaybreaks
\begin{document}
\bstctlcite{IEEEexample:BSTcontrol}
\maketitle
\begin{abstract} 

Cell-free massive multiple-input multiple-output  is a potential candidate for future networks with pervasive connectivity by utilizing coherent joint transmission and distributed antenna arrays. This paper studies the exploitation of full-duplex communication for a distributed antenna array. Specifically, we derive a closed-form expression for the uplink and downlink ergodic spectral efficiency (SE) for a network where the APs can flexibly operate in either the full-duplex or half-duplex mode with linear processing and Rayleigh fading channels. A long-term total SE maximization problem is formulated subject to a network operation model and individual SE requirements with limited power budget. Due to the intrinsic nonconvexity and infeasible circumstances where some UEs might not be able to achieve the rate requirements, we adapt differential evolution to design a low computational complexity algorithm that can attain good power allocation and network operation mode in polynomial time. Numerical results demonstrate the effectiveness of our system design and proposed algorithm over state-of-the-art benchmarks with satisfactory service to the majority of UEs, although several ones may be unscheduled under harsh conditions.
\end{abstract}

\begin{IEEEkeywords}
	Cell-free massive multiple-input multiple-output (CF-mMIMO), spectral efficiency (SE),  network-assisted full-duplex (NAFD), differential evolution.
\end{IEEEkeywords}

\section{Introduction}
Mobile data traffic continues to grow, driven by the proliferation of Internet of Things (IoT) devices, immersive applications, such as augmented and virtual reality (AR/VR), and the demand for ubiquitous connectivity~\cite{luo2025wireless}. Although fifth-generation (5G) networks have improved mobile broadband services, they still face several challenges, including costly infrastructure deployment, limited millimeter-wave coverage, and scalability issues in ultra-dense scenarios. To overcome these limitations, sixth-generation (6G) networks are envisioned to deliver higher spectral and energy efficiency, ultra-reliable low-latency communications, and massive device connectivity \cite{paper3_ji2021survey}. Beyond enhanced broadband, 6G is expected to enable a wide range of emerging applications, such as large-scale digital twins \cite{paper1_zheng2025cognitive}, autonomous vehicle networks, smart city infrastructures with satellite–terrestrial integration \cite{paper2_zhang2024cost}, and efficient communication through reconfigurable intelligent surfaces \cite{chien2024active}. These capabilities aim to support seamless connectivity and enable new services,  contributing to the evolution of the digital ecosystem, enhancing user experiences, and facilitating future technological developments. In this context, cell-free massive multiple-input multiple-output (CF-mMIMO) and full-duplex (FD) communications have emerged as promising enablers to underpin these demanding requirements~\cite{kurma2024performance}.

CF-mMIMO has been introduced as a solution to mitigate inter-cell interference and to offer balanced quality-of-experience services for users (UEs)~\cite{elhoushy2021cell}. In CF-mMIMO, a large number of access points (APs) antennas are distributed within the coverage area and coordinated by several data center-level computing units (DCCUs) to simultaneously serve many UEs in the same time-frequency resources. The primary distinction of CF-mMIMO, in contrast to conventional cellular mMIMO, lies in the expectation of a significantly greater number of APs compared to UEs, operating within a regime devoid of cell boundaries \cite{ngo17TWC,emil20TWC}. CF-mMIMO inherits key benefits of collocated massive MIMO (Co-mMIMO), such as channel hardening and favorable propagation \cite{marzetta2010noncooperative}, while offering superior spectral efficiency (SE) over small-cell and Co-mMIMO networks through joint/coherent processing of many APs at the DCCUs.

FD communication in CF-mMIMO represents an advanced  mechanism by enabling simultaneous transmission and reception on the same frequency channel \cite{hieu20JSAC,anokye2023power}. By enabling uplink (UL) and downlink (DL) transmissions to happen simultaneously on the same time-frequency resources, FD systems can reclaim the bandwidth that is typically lost in traditional half-duple (HD) systems. However, interference management in FD CF-mMIMO systems is more complex than in their HD counterparts due to the presence of residual self-interference (RSI) that persists even after self-interference (SI) suppression. Although passive and active SI suppression techniques have effectively reduced the SI power to background noise levels, deploying power-hungry SI suppression circuits across all APs would increase the power consumption to unacceptable levels, especially given the large number of APs in the network~\cite{smida2023full}. Additionally, inter-AP/UE interference, known as cross-link interference (CLI), degrades FD CF-mMIMO network performance. The rising CLI impacts the SE and energy efficiency, often undermining the potential gains over HD systems, even with the joint DL/UL power allocation and large-scale fading decoding (LSFD) designs.

To mitigate the adverse impacts of RSI and CLI on the performance of FD CF-mMIMO networks, an effective solution is to adopt a novel physical layer design for CF-mMIMO, termed network-assisted full duplexing (NAFD), as introduced in~\cite{Wang:TCOM:2020}.  In NAFD CF-mMIMO, APs can function in FD mode (with all APs handling UL and DL simultaneously on the same frequency), hybrid-duplex mode (with a mix of HD and FD APs in the network), or flexible-duplex mode (where APs operate in HD). This allows NAFD to integrate all duplex modes within the network~\cite{Wang:TCOM:2020,Jiamin:TWC:2021}. Each AP's transmission mode (i.e., UL reception and/or DL transmission) can be dynamically adjusted according to network conditions and requirements (e.g., the location and number of UL and DL UEs, as well as their quality-of-service (QoS) requirements), enabling simultaneous and adaptable support for both UL and DL services. Consequently, the need for FD operation at all APs is reduced, as only some APs (if required)  will operate in FD mode. This flexibility decreases the overall network energy consumption due to the reduced need for SI suppression. Meanwhile, the UL and DL operation modes of the HD APs are determined by the network's configuration. Specifically, APs near UEs engaged in UL transmission will perform UL reception, while those near UEs with DL traffic will switch to DL transmission. This approach helps manage the amount of CLI in the network.  

\subsection{Review of Related Literature}
 NAFD CF-mMIMO systems have been studied from multiple perspectives and across various configurations~\cite{Wang:TCOM:2020,Xia:TVT:2020,Jiamin:TWC:2021,Xinjiang:TWC:2021,Zhu:COML:2021,Xia:China:2021,Xia:SYSJ:2022,Chowdhury:TCOM:2024,Mohammad:JSAC:2023}. Wang~\textit{et al.}~\cite{Wang:TCOM:2020} first proposed the idea of NAFD CF-mMIMO, inspired by the COMPflex concept introduced in~\cite{Thomsen:WCL:2016}. They analyzed the SE of NAFD CF-mMIMO with HD APs, considering the impact of imperfect CSI and spatial correlations.  Li~\textit{et al.}~\cite{Jiamin:TWC:2021} proposed a beamforming training approach to cancel inter-AP interference and facilitate coherent decoding at the DL terminals. The authors also developed a multi-objective optimization problem aimed at minimizing the total power used for pilot and data transmission, while adhering to per-terminal QoS and power budget constraints. Xia~\textit{et al.}~\cite{Xia:TVT:2020} developed a two-stage iterative semi-definite relaxation and block coordinate descent algorithm to address the highly non-convex joint transceiver design problem. The idea in ~\cite{Xinjiang:TWC:2021} simultaneously optimized UE selection, fronthaul compression ratio,  and transceiver beamforming to maximize the SE and the number of UEs admitted to the network.
 
 The studies in~\cite{Wang:TCOM:2020,Xia:TVT:2020,Jiamin:TWC:2021,Xinjiang:TWC:2021} assumed a fixed mode assignment at the APs, whereas the authors in~\cite{Zhu:COML:2021,Xia:China:2021,Xia:SYSJ:2022,Chowdhury:TCOM:2024,Mohammad:JSAC:2023} demonstrated that the performance of NAFD CF-mMIMO networks can be significantly enhanced through dynamic AP mode assignment. Zhu \textit{et al.}~\cite{Zhu:COML:2021} presented a parallel successive convex approximation (SCA) method to solve the complex optimization problem for duplex mode selection. To further enhance efficiency, they later devised a refined reinforcement learning algorithm based on Q-learning, aimed at reducing computational demands. Xia \textit{et al.}~\cite{Xia:China:2021} investigated duplex mode selection and transceiver design to maximize the aggregated SE of DL and UL, considering the per-UE QoS constraints and power budget constraints at APs and UEs. In~\cite{Xia:SYSJ:2022}, the problem of AP mode selection and DL beamforming design was studied from a secrecy perspective, where artificial noise was injected into the information signals at the APs to prevent interception by eavesdroppers.  Chowdhury \textit{et al.}~\cite{Chowdhury:TCOM:2024} proposed a heuristic approach for AP mode selection and UL/DL power control to maximize the aggregate SE of UL and DL transmissions. Mohammadi~\textit{et al.}~\cite{Mohammad:JSAC:2023} developed a general SCA-based optimization framework to maximize the SE and energy efficiency of the NAFD CF-mMIMO systems, accounting for per-UE QoS requirements, backhaul power consumption, and power constraints at the UEs and APs.

\vspace{-0.5em}
\subsection{Research Gap and Main Contributions}
A common characteristic of~\cite{Xia:TVT:2020,Xinjiang:TWC:2021,Zhu:COML:2021,Xia:China:2021,Xia:SYSJ:2022} is that, from the perspective of AP mode selection and power control design, these works rely heavily on instantaneous channel state information (CSI). This reliance requires that, during each coherence interval and following the channel estimation phase, all acquired CSI is transmitted to the DCCUs for centralized decision-making. Such an approach introduces significant signaling overhead, especially in CF-mMIMO systems with large numbers of APs and UEs. To address these challenges, \cite{Wang:TCOM:2020,Jiamin:TWC:2021,Chowdhury:TCOM:2024,Mohammad:JSAC:2023} have leveraged statistical CSI for system design. In particular, \cite{Wang:TCOM:2020} utilized large-dimensional random matrix theory to derive deterministic equivalents for their UL and DL sum-rate analysis. In \cite{Jiamin:TWC:2021}, the authors proposed estimating the effective DL CSI, which is defined as the inner product between the beamforming and channel vectors, to enable reliable signal decoding under reduced channel hardening. The statistical CSI-based design in \cite{Chowdhury:TCOM:2024}, however, relies on the unrealistic assumption of inter-AP CSI availability. Nevertheless, the design frameworks proposed in \cite{Wang:TCOM:2020, Jiamin:TWC:2021, Chowdhury:TCOM:2024} do not consider per-user QoS requirements, which may lead to service outages for some users. Additionally, the SCA-based design in~\cite{Mohammad:JSAC:2023} faces scalability challenges in large and dense networks and requires lengthy computation times, rendering it impractical for such environments. Another limitation shared by all the aforementioned works, except for~\cite{Xia:China:2021}, is the adoption of a flexible-duplex prototype within the network. The APs are restricted to HD operation, which constrains their ability to support more dynamic and efficient communication modes.  The proposed design primarily relies on statistical CSI. Given that large-scale channel coefficients are updated far less frequently than their small-scale fading counterparts, typically once every 20 to 40 coherence intervals, the network configuration remains fixed over multiple coherence intervals. In such cases, asymmetry between the downlink and uplink data loads may lead to suboptimal performance for some UEs. In practical scenarios, uplink traffic is generally lighter than downlink traffic, and some users may switch their demand from uplink to downlink before the next network update. By accounting for this potential asymmetry and incorporating a degree of tolerance into the network design, we can continue to serve these users effectively without requiring an immediate reconfiguration of the system. 

To address these bottlenecks, we investigate a dynamic NAFD CF-mMIMO network in which subsets of APs can operate in either HD or FD mode to enhance the SE requirements for both UL and DL UEs. The optimal subsets of APs are determined, to maximize the aggregate UL and DL SE under given QoS requirements at all UEs and power constraints at the APs and UEs. To provide a flexible trade-off between interference cancelation and signal boosting, while adhering to system constraints on the available transmit and receive antennas, we implement local partial zero-forcing (PZF) precoding and decoding designs for DL and UL communications, respectively. Using PZF design, each AP adopts zero-forcing (ZF) and  maximum-ratio (MR) linear processing to support UL and DL transmissions of different groups of UEs. The principle of this scheme is that APs in the DL (or UL) suppress interference primarily for (from) the strongest UEs— those with the highest channel gain and most likely to experience (provide) significant interference. In contrast, interference to (from) the weakest UEs is tolerated~\cite{Interdonato:TWC:2020}. Our main contributions are summarized as follows:
\begin{itemize}
\item We derive the UL and DL SEs of each UE for the predetermined operating modes of APs and linear signal processing PZF. The analytical SE expressions unveil the influence of practical aspects including large-scale fading and imperfect CSI on the performance. 
\item We formulate a total UL and DL SE maximization problem subject to the limited transmit power budgets, the network operating modes, and the given individual QoS requirements for UEs. Even though the considered problem is non-convex and NP-hard, it is a generalized version of most previous works in the literature~\cite{Mohammad:JSAC:2023}, where MR processing with HD APs is considered.   
\item An adapted differential evolution (DE)-based algorithm is proposed to handle the inherent nonconvexity and obtain the optimal solution in polynomial time. In infeasible circumstances where all UEs' requirements cannot be simultaneously met, our proposed algorithm can identify unsatisfied UEs while providing the required services to the remaining ones. 
\item Numerical results show the effectiveness of the NAFD CF-mMIMO systems in enhancing the SE of each UE. Furthermore, the optimized modes of APs significantly improve the network services and handle the congestion.
\end{itemize}

\textit{Notation:} We use bold upper case letters to denote matrices, and lower case letters to denote vectors. The superscripts $(\cdot)^*$, $(\cdot)^T$ and $(\cdot)^\dag$ stand for the conjugate, transpose, and conjugate-transpose, respectively.   The circular symmetric complex Gaussian distribution having variance $\Sn$ is denoted by $\mathcal{CN}(0,\Sn)$. Finally, $\Ex\{\cdot\}$ is the statistical expectation.

\vspace{-0.5em}
\section{System model}~\label{sec:Sysmodel}
An NAFD CF-mMIMO system is considered, where $M$ APs coherently serve $K_u$ UL UEs and $K_d$ DL UEs. We denote $\mathcal{M}\triangleq\{1,\ldots,M\}$, $\K_d\triangleq \{1,\dots,K_d\}$ and  $\K_u\triangleq\{1,\ldots,K_u\}$  the sets of indices of the APs, DL UEs, and UL UEs, respectively. Moreover,  $\UEul$ and $\UEdk$ refer to UL UE $\ell$ and the DL UE $k$, respectively. All APs are connected to the DCCU via high-capacity backhaul links. The DL and UL data transmissions are performed simultaneously exploiting the same time and frequency band via HD and FD APs. Each UE has one single antenna. Meanwhile, each AP has  $N$ transmit radio frequency (RF) chains and $N$ receive RF chains. A coherence interval consists of UL pilot training for channel estimation followed by DL/UL data transmission. The APs operate in either HD (UL or DL) or FD (simultaneous UL and DL), depending on network dynamics and requirements.

\subsection{UL Pilot Training and Channel Estimation}
\label{phase:ULforCE}
The propagation channel vector between the $\UEdk$ ($\UEul$) and the $m$-th AP, $\gmkd\in\mathbb{C}^{\Ntx \times 1}$ ($\gmlu\in\mathbb{C}^{\Nrx \times 1}$) is modelled as $\gmkd=\sqrt{\betamkd}\tgmkd,~(\gmlu=\sqrt{\betamlu}\tgmlu) $, where $\betamkd$ ($\betamlu$) is the large-scale fading coefficient, while $\tgmkd\in\mathbb{C}^{\Ntx \times 1}$ ($\tgmlu\in\mathbb{C}^{\Ntx \times 1}$) is the small-scale fading vector with the elements being independent and identically distributed (i.i.d.) $\mathcal{CN} (0, 1)$. Moreover, the channel coefficient between the  $\UEul$ to the $\UEdk$ is denoted as $h_{kl}\sim\CN(0,\betakldu)$, where $\betakldu$ is the large-scale fading coefficient.  Finally, we denote the channel matrix from AP $m$ to AP $i$, for $i \neq m$, as $\qF_{mi} \in \mathbb{C}^{\Nrx \times \Ntx}$,  whose elements  are i.i.d. $\mathcal{CN}(0,\beta_{mi})$ RVs. For $i=m$, $\F_{mm}, \forall m$, represents the SI channel at the FD APs, with elements that are i.i.d. $\mathcal{CN}(0,\SIm)$.


Within a coherence interval of $\tau_c$  symbols, each UL/DL UE transmits a pairwise orthogonal pilot sequence of 
$\tau_t$ symbols to all APs. This condition necessitates $\tau_t\geq K_d + K_u$. Relying on the received pilot signals at AP $m$ and by using the minimum mean-square error (MMSE) estimation technique, $\gmkd$  and $\gmlu$ are estimated. By following~\cite{ngo17TWC},  the MMSE estimates $\hgmkd$ and $\hgmlu$ of $\gmkd$  and $\gmlu$ are modeled as $\hgmkd \sim \mathcal{CN}(\boldsymbol{0},\gamdmk \mathbf{I})$, and $\hgmlu \sim \mathcal{CN}(\boldsymbol{0},\gamuml \mathbf{I})$, respectively, where $\gamdmk \triangleq \frac{{\tau_t\rho_t}(\betamkd)^2}
{\tau_t\rho_t
\betamkd+1}, 
\gamuml \triangleq 
\frac{{\tau_t\rho_t}(\betamlu)^2}
{\tau_t\rho_t
\betamlu
+1}$.

\vspace{-1em}
\subsection{DL-and-UL Payload Data Transmission}
Prior to this phase, the operating mode of the APs—whether UL reception and/or DL transmission—is determined according to the proposed design algorithm (see Section~\ref{sec:SE}). It is important to note that the selection of AP modes is carried out based on the large-scale fading timescale, which evolves gradually over time. The binary variables representing the mode assignment for each AP $m$ are defined as follows
\vspace{-0.3em}
\begin{align}
\label{a}
a_{m} (b_m) \triangleq
\begin{cases}
  1, & \text{if AP $m$ operates in the DL (UL) mode,}\\
  0, & \mbox{otherwise}.
\end{cases}
 \end{align}
Here, for the HD AP $m$, we have
\vspace{-0.3em}
\begin{align}
\label{sumabHD}
    a_m + b_m = 1, ~\text{or}~ a_mb_m=0,~\forall m
\end{align}
which guarantees AP $m$ only operates in HD mode. Moreover, for the FD AP $m$, we have\footnote{Although an FD AP can, in principle, be represented by two HD APs, one dedicated to UL and the other to DL, our study focuses on a fixed network deployment where both the number and locations of APs are predetermined. Within this framework, the only available design degree of freedom is to dynamically assign the operation mode of each AP (FD, HD-UL, or HD-DL), which plays a critical role in avoiding infeasibility and ensuring robust network performance.} 
\vspace{-0.3em}
\begin{align}
\label{sumabFD}
    a_m + b_m = 2,~\text{or}~ a_mb_m=1,~ \forall m.
\end{align}
\subsubsection{DL payload data transmission}
Using the locally acquired channel estimates at each AP, the APs apply linear processing techniques, such as MR or ZF, to the signals transmitted to the $K_d$ DL UEs. Let $s_k^\dl\sim\mathcal{CN}(0,1)$ denote the intended symbol for $\UEdk$. The AP $m$ uses the precoding vector $\wmkdl \in \mathbb{C}^{\Ntx \times 1}$ and transmits
\vspace{-0.3em}
\begin{align}~\label{eq:x_m}
 \qx_{m}^{\dl}
= a_m\sqrt{\rho_d}\sum\nolimits_{k \in \mathcal{K}_d} \theta_{mk} \wmkdl
s_{k}^{\dl},   
\end{align}
where $\rho_d$ is the maximum normalized transmit power at each AP  and $\theta_{mk} \geq 0$, for all $m$ and $k$, is the power control coefficient at AP $m$ for UE $k$. Note that $a_m$ in~\eqref{eq:x_m} clearly determines whether AP $m$ is operating in DL or not.

The received signal at $\UEdk$ is written as
\vspace{0.1em}
\begin{align}~\label{eq:ykdl}
y_k^{\dl}
&=
\sqrt{\rho_d}\sum\nolimits_{m \in \mathcal{M}} a_m\theta_{mk}
\left(\gmkd\right)^\dag\wmkdl
s_{k}^{\dl}
\nonumber\\
&\hspace{0em}+
\sqrt{\rho_d}
\sum\nolimits_{m \in \mathcal{M}}
\sum\nolimits_{k'\in\mathcal{K}_d \setminus k} 
a_m\theta_{mk'}
\left(\gmkd\right)^\dag \wmkpdl
s_{k'}^{\dl}
\nonumber\\
&\hspace{0em}+
\sum\nolimits_{\ell\in \mathcal{K}_{u}}h_{k\ell}x_{\ell}^{\ul}+w_{k}^{\dl},
\end{align}
where $w_{k}^{\dl}\sim\mathcal{CN}(0,1)$ is the AWGN at $\UEdk$. The third term in~\eqref{eq:ykdl}, denotes the CLI caused by the UL UEs transmitting simultaneously over the same frequency band. Moreover, $x_\ell^\ul$ represents the transmit signal from the $\UEul$. 

\subsubsection{UL payload data transmission}
Let $s_{\ell}^\ul$, with $\Ex\left\{|s_{\ell}^\ul|^2\right\}=1$ denote the transmit symbol by $\UEul$ and let $\rho_u$ represent  the maximum normalized transmit power at each UL UE. By considering the UL power control at each UL UE,  the transmit signal  from  $\UEul$ is given by $x_{\ell}^\ul  = \sqrt{\rho_u {\varsigma}_\ell} s_{\ell}^{\ul}$, where ${\varsigma}_{\ell}$ is the transmit power control coefficient at $\UEul$ with
\vspace{-0em}
\begin{align}
\label{UL:power:cons}
    0\leq {\varsigma}_{\ell} \leq 1, \forall \ell.
\end{align}
The APs operating in UL mode (i.e., all APs with  $b_m=1, \forall m$), receive the signal transmitted from all UL UEs. Thus, the received signal $\qy_{m}^{\ul}\in\mathbb{C}^{\Nrx \times 1}$ at AP $m$ can be written as
\begin{align}\label{eq:ymul}
\qy_{m}^{\ul}
&=
\sqrt{\rho_u}\sum\nolimits_{\ell\in \mathcal{K}_{u}}b_m\sqrt{ {\varsigma}_{\ell}}\qg_{m\ell}^{\ul} s_{\ell}^{\ul}
\nonumber\\
&
+
\sqrt{\rho_d}\sum\nolimits_{i\in\mathcal{M}\setminus m}\sum\nolimits_{k\in \mathcal{K}_d}
{b_m a_i } \theta_{ik}
\qZ_{mi}
\wikdl s_k^\dl\nonumber\\
&
\hspace{0em}
\!+\!
\sqrt{\rho_d}\sum\nolimits_{k\in \mathcal{K}_d}\!\!\!\!
{b_m a_m } \theta_{mk}
\qZ_{mm}
\wmkdl s_k^\dl
\!+\!{b_m}\qw_{m}^{\ul},
\end{align}
where $\qw_{m}^{\ul}$ is the $\mathcal{CN}(\boldsymbol{0},\qI_N)$ AWGN vector. In~\eqref{eq:ymul}, the second term captures the interference from APs transmitting towards DL UEs and the third term represents the SI term in case of AP $m$ is FD-enabled (i.e., $b_m a_m=1$). \eqref{eq:ymul} indicates that when AP $m$ is not operating in UL mode $\qy_{m}^{\ul}=\boldsymbol{0}$ as $b_m=0$. 

Then, AP $m$ performs linear processing to the received signal in~\eqref{eq:ymul}. Let $\wmlul\in\mathbb{C}^{N\times 1}$ denote the linear precoding vector at AP $m$. The resulting $(\wmlul)^\dag\qy_{m}^{\ul}$  is then sent to the DCCU for signal detection. To enhance the achievable UL SE, the DCCU further multiplies the received signal by the LSFD coefficient $\alphml, \forall m,k$. Then, the aggregated received signal for $\UEul$, at the DCCU can be expressed as~\cite{Bashar:TWC:2019}
\begin{align}\label{eq:rul}	\qr_{\ell}^{\ul}=\sum\nolimits_{m\in\MM}\alphml(\wmlul)^\dag\qy_{m}^{\ul}, \quad \forall \ell\in\K_u. 
\end{align}
Without loss of generality, we
assume that
\begin{align}\label{eq:alphml}
|\alphml|^2\leq 1,  \quad \forall \ell, m.
\end{align}

\vspace{-2em}
\subsection{Precoding/Combining Design}

\subsubsection{Local Partial Zero-Forcing Precoding}
We consider the local  PZF precoding scheme, which effectively mitigates interference in a distributed and scalable way, providing a flexible trade-off between large array gain and interference reduction~\cite{Interdonato:TWC:2020}.\footnote{The PZF design effectively overcomes the scalability limitations of MR (which arise from inter-UE interference) and the implementation challenges of ZF (i.e., the number of antennas at each AP must satisfy the condition $N>K_d$) within a distributed framework. To successfully eliminate interference among the APs and UEs, a centralized precoding and decoding scheme can be employed in the network. In this approach, the DCCU performs the design based on the complete CSI between all UEs and all APs. While this global CSI availability can lead to enhanced performance, it also imposes a substantial fronthaul burden, as all channel estimates must be transmitted to the DCCU during each coherence interval. We consider the exploration of such a centralized design an important direction for future research.} To this end, each AP $m$ in DL mode (FD AP or HD AP in DL mode) divides the DL UEs into two groups: $\Sm \subset \{1, \ldots, K_d\}$, which includes the index of strong DL UEs, and $\Wm \subset \{1, \ldots, K_d\}$, which includes the index of weak DL UEs, respectively. Then, AP $m$ employs the ZF precoding for DL UEs in $\Sm$  and MR precoding for DL UEs in $\Wm$. 
The local PZF transmit precoding at AP $m$, is given by  $\wmkdlzf = \gamdmk \hGmdl\big( \big(\hGmdl\big)^\dag \hGmdl\big)^{-1} \qe_k$, where
$\hGmdl$ is an $N \times |\Sm|$ collective channel estimation matrix from all the UEs in $\Sm$ to AP $m$ as $\hGmdl=[\hat{\qg}_{mk}^\dl: k \in \Sm]$. Moreover, $\qe_ k$ is the $k$-th column of $\qI_{K_d}$. Therefore, for any pair of $\UEdk$ and  $\UEdkp$ in $\Sm $, we have  
\vspace{-0.5em}
\begin{align}\label{eq:PZF_prec2}
(\hgmkpd)^\dag \wmkdlzf = 
\begin{cases} \displaystyle \gamdmk & \text {if } k=k',\\ \displaystyle 0 & \text {otherwise}.
\end{cases}
\end{align}
Furthermore,  AP $m$ forms the MR precoding vector locally for $\UEdk$, with $k \in \Wm$, is given in by  $\wmkdlmr = \hgmkd$.

To represent the DL group assignment in our PZF combining scheme, we define a pair of binary variables for each $\UEdk$ and DL AP $m$ as follows
\begin{align*}
\dm^{\Zk} = \begin{cases} \displaystyle 1 & \text {if } m \in \Zk,\\ \displaystyle 0 & \text {otherwise},
\end{cases}
\qquad
\dm^{\Tk} = \begin{cases} \displaystyle 1 & \text {if } m \in \Tk,\\ \displaystyle 0 & \text {otherwise},
\end{cases}
\end{align*}
where  $\Zk$ and $\Tk$ denote the set of indices of APs that assign the $\UEdk$ into $\Sm$ for  ZF precoding and the set of indices of APs that assign the $\UEdk$ into $\Wm$ for MR precoding, respectively, defined as
\begin{align}
\Zk \triangleq \{m: k \in \Sm, m=1, \ldots, M\},\\
\Tk \triangleq \{m: k \in \Wm, m=1, \ldots, M\},
\end{align}
with $\Zk\cap \Tk =\emptyset$ and $\Zk\cup \Tk = \mathcal{M}^{\dl}\subset \mathcal{M}$, where $\mathcal{M}^{\dl}$ includes all FD APs and HD APs operating in DL.

With PZF design, the transmit signal in~\eqref{eq:x_m} is given by
\begin{align}~\label{eq:x_m:pzf}
 \qx_{m}^{\dl}
&= a_m\sqrt{\rho_d}\Big(\sum\nolimits_{k \in \Sm} \theta_{mk} \wmkdlzf s_{k}^{\dl}
\nonumber\\
&\hspace{5em}
 + \sum\nolimits_{k \in \Wm} \theta_{mk}
\wmkdlmr s_{k}^{\dl}\Big).
\end{align}
Note that AP $m$ must satisfy the average normalized power constraint, i.e., $\Ex\left\{\|\qx_{m}^{\dl}\|^2\right\}\leq \rho_d$.  Thus, we derive the following power constraint for each AP~\cite{ngo17TWC}
\begin{align}
\label{DL:power:cons2}
\sum\nolimits_{k\in\K} \gamdmk \varphi_{mk}\theta_{mk}^2 \leq a_m,
\end{align}
where $\varphi_{mk} \triangleq \frac{\dm^{\Zk} }{N-\vert\Sm\vert} +
 \dm^{\Tk}  N$, and we have used the fact that $\Ex\big\{\big\Vert \wmkdlzf\big\Vert^2\big\} =\frac{\gamdmk}{N-\vert\Sm\vert}$.

\subsubsection{Local Partial Zero-Forcing Combining}
Similar to the DL precoding design, we assume for the UL combining vector design that each AP $m$, whether operating in FD mode or as a HD AP in UL mode, categorizes the UL UEs into two groups: 
$\Smu \subset \K_u$, representing the indices of strong UL UEs, and 
$\Wmu\subset \K_u$, representing the indices of weak UL UEs. The local PZF combining at AP $m$ is designed to suppress the interference from strong UL UEs by applying ZF combining. Conversely, the interference generated by weak UL UEs which have weak channel gains, is tolerated through the use of MR combining.  The ZF combining at AP $m$, is given by  $\wmlulzf = \gamuml\hGmul\big( \big(\hGmul\big)^\dag \hGmul\big)^{-1}\qe_{\ell}$, where $\qe_{\ell}$ is the $\ell$-th column of $\qI_{K_u}$ and $\hGmul = [\hgmonu,\ldots, \hgmKuu]$. Moreover,  for MR combining vector, we set $\wmlulmr=\hgmlu$.

To represent the UL group assignment in our PZF combining scheme, we define a pair of binary variables for each $\UEul$ and UL AP $m$ as follows
\begin{align*}
\dm^{\Zlup} = \begin{cases} \displaystyle 1 & \text {if } m \in \Zlup,\\ \displaystyle 0 & \text {otherwise},
\end{cases}
\qquad
\dm^{\Tlup} = \begin{cases} \displaystyle 1 & \text {if } m \in \Tlup,\\ \displaystyle 0 & \text {otherwise},
\end{cases}
\end{align*}
where $\Zlup$ ($\Tlup$) denotes the set of indices of APs that assign $\UEul$ into $\Smu$ for ZF combining (the set of indices of APs that assign $\UEul$ into $\Wmu$ for MR combining), with
\begin{align}
\Zlup &\triangleq \{m: \ell \in \Smu, m=1, \ldots, M\},\\
\Tlup &\triangleq \{m: \ell \in \Wmu, m=1, \ldots, M\},
\end{align}
where $\Zlup\cap \Tlup =\emptyset$ and $\Zlup\cup \Tlup = \mathcal{M}^{\ul}\subset \mathcal{M}$, while $\mathcal{M}^{\ul}$ consists of all FD APs and HD APs operating in UL.

\subsection{DL/UL UE Grouping}
The  DL/UL UE grouping can be based on different criteria.  Our approach to UL and DL UE grouping is inspired by \cite[Eq. (44)]{Hien:TGCN:2018} and \cite[Eq. (44)]{Interdonato:TWC:2020}. This method follows a heuristic, distributed, channel-dependent power control policy, and is based on the following rule
\begin{align}~\label{eq:groupin:criterion}
   \sum\nolimits_{k=1}^{\vert\hat{\mathcal{S}}_m \vert} \frac{\bar{\beta}_{m,k}}{\sum_{t=1}^{K_d} \beta_{m,t}} \geq \upsilon\%,
\end{align}
where according to which AP $m$ defines its set $\hat{\mathcal{S}}_m$ by including UEs whose contribution to the total channel gain is at least $\upsilon\%$. In~\eqref{eq:groupin:criterion}, $\{\bar{\beta}_{m,1},\ldots,\bar{\beta}_{m,K_d}\}$ represents the collection of large-scale fading coefficients sorted in decreasing order.

\vspace{-0.1em}
 \section{SE Analysis and Problem Formulation}~\label{sec:SE}
\vspace{-2.5em}
\subsection{Downlink Spectral Efficiency}
In order to detect $s_{k}^{\dl}$, the $\UEdk$ needs to have access to the effective channel. However, since no pilot is transmitted in the DL, this CSI is not available at $\UEdk$. To address this issue, we apply the use-and-then-forget capacity-bounding technique~\cite{ngo17TWC}, where $\UEdk$ uses the stochastic CSI to detect $s_{k}^{\dl}$. Therefore, by applying local PZF precoding, we can rewrite the received signal in~\eqref{eq:ykdl} as
\begin{align}~\label{eq:ykdl:detect:PZF}
y_k^{\dl}
&=\mathbb{DS}_k^\PZF s_{k}^{\dl} + \mathbb{BU}_k^\PZF s_{k}^{\dl}
+\sum\nolimits_{k'\in\mathcal{K}_d \setminus k} \mathbb{DI}_{kk'}^\PZF s_{k'}^{\dl}
\nonumber\\
&\hspace{2em}
+\sum\nolimits_{\ell\in \mathcal{K}_{u}} \mathbb{UI}_{\ell}^\PZF s_{\ell}^{\ul}
+w_{k}^{\dl},
\end{align}
where $\mathbb{UI}_{\ell}^\PZF\triangleq
h_{k\ell}\sqrt{\rho_u \varsigma_\ell}$ and
\begin{align}
\mathbb{DS}_k^\PZF &\triangleq\sqrt{\rho_d}\Ex\Big\{\sum\nolimits_{m \in \Zk} a_m\theta_{mk}
\left(\gmkd\right)^\dag\wmkdlzf
\nonumber\\
&\hspace{2em}
+
\sum\nolimits_{m \in \Tk} a_m\theta_{mk}
\left(\gmkd\right)^\dag\wmkdlmr
\Big\},
\nonumber\\
\mathbb{BU}_k^\PZF &\triangleq\sqrt{\rho_d}
\Big(\sum\nolimits_{m \in \Zk} a_m\theta_{mk}
\left(\gmkd\right)^\dag\wmkdlzf
\nonumber\\
&\hspace{2em}
+
\sum\nolimits_{m \in \Tk} a_m\theta_{mk}
\left(\gmkd\right)^\dag\wmkdlmr
\Big) - \mathbb{DS}_k 
\nonumber\\
\mathbb{DI}_{kk'}^\PZF &\triangleq
\sqrt{\rho_d
}\sum\nolimits_{k'\in\mathcal{K}_d \setminus k}
\Big(
\sum\nolimits_{m \in \Zk}
a_m\theta_{mk'}
\left(\gmkd\right)^\dag \wmkpdlzf
\nonumber\\
&\hspace{2em}
+
\sum\nolimits_{m \in \Tk}
a_m\theta_{mk'}
\left(\gmkd\right)^\dag \wmkpdlmr
\Big),
\end{align}
represent the strength of the desired DL signal ($\mathbb{DS}_k$), the beamforming
gain uncertainty ($\mathbb{BU}_k$), CLI caused by the $\UEdkp$ ($\mathbb{DI}_{kk'}$), and CLI caused by the $\UEul$ ($\mathbb{UI}_{\ell}$), respectively. The combined contribution of the last four terms in~\eqref{eq:ykdl:detect:PZF} is treated as effective noise, which remains uncorrelated with the target signal~\cite{ngo17TWC}.

\begin{proposition}~~\label{Prop:SE:DLPZF}
With PZF precoding, the achievable DL SE at $\UEdk$ is given by 
	\begin{align}~\label{eq:DL:SE}
\mathcal{S}_{\dl,k}^{\PZF} (\qa, \boldsymbol \theta, {\boldsymbol{\varsigma}}) =  \frac{\tau_c-\tau_t}{\tau_c}
\log_2 \left(1  
    + \frac{(\Xi_k^\PZF(\boldsymbol \theta, \qa))^2}{\Omega_k^\PZF(\boldsymbol \theta, {\boldsymbol{\varsigma}}, \qa)}
	\right),
	\end{align}
where $\qa \triangleq \{a_m\}$, $\boldsymbol{\theta}\triangleq \{\theta_{mk}\}$, ${\boldsymbol{\varsigma}} \triangleq \{{\varsigma}_{\ell}\}, \forall m,k,\ell$ and
\begin{align}
    \nonumber
    &\Xi_k^\PZF (\boldsymbol \theta, \qa) \triangleq 
    \sqrt{\rho_d}\sum\nolimits_{m \in \MM}\Big( \dm^{\Zk} a_m\theta_{mk} \gamdmk
    \nonumber\\
    &\hspace{6em}
\!+\!
\dm^{\Tk}N  a_m\theta_{mk} \gamdmk\Big)
    \\
    \nonumber
    &\Omega_k^\PZF (\boldsymbol \theta, {\boldsymbol{\varsigma}}, \qa) \triangleq
    {\rho_d}
    \sum\nolimits_{m \in \MM} 
    \sum\nolimits_{k'\in\mathcal{K}_d}    
    \Big(\dm^{\Zk} \!\!a_m(\betamkd-\gamdmk)
    \\
    \nonumber
    &\times\frac{\theta_{mk'}^2 \gamdmkp}{N \!-\! \vert \Sm\vert} 
\!+\!
 \dm^{\Tk} N  a_m\theta_{mk'}^2
\gamdmkp\betamkd \Big) 
\!+\!\rho_u\!\!\sum\nolimits_{\ell\in\K_u}
\!\!\!\varsigma_\ell\betakldu\!+\!1.
\end{align}
\end{proposition}

\begin{proof}
See Appendix~\ref{Proof:Prop:SE:DLPZF}.
\end{proof}

\begin{Remark}
The DL SE result in Proposition~\ref{Prop:SE:DLPZF} is comprehensive and encompass both full-ZF (FZF) and MR precoding as special cases: 
\vspace{-0.2em}
\begin{itemize}
    \item FZF precoding:  $\dm^{\Zk}=1$, $\dm^{\Tk}=0$, and $\vert \Sm \vert =K_d$.
    \item MR precoding: $\dm^{\Zk}=0$ and $\dm^{\Tk}=1$,
\end{itemize}
\label{remark_1}
\end{Remark}

\vspace{-1em}
\subsection{Uplink Spectral Efficiency}
In order to detect the desired signal $s_{\ell}^{\ul}$ from $	\qr_{\ell}^{\ul}$ at the DCCU, it relies only on statistical knowledge of the channels. To this end, we employ the use-and-then-forget capacity-bounding technique~\cite{ngo17TWC} and obtain the achievable UL SE (in bit/s/Hz) of the $\UEul$. We notice that according to~\eqref{eq:rul}, the received signal at the DCCU is a function of $\wmlul$ and $\wmkdl$. Therefore, different choices of DL and UL beamforming designs at the APs, result in different expressions for the UL SINR. For general case of local PZF combining and PZF precoding (PZF/PZF), the received UL signal~\eqref{eq:rul}, can be expressed as
\begin{align}\label{eq:rul_pzf}
r_{\ell}^{\PZF}&=
\DSlpzf  s_{\ell}^\ul
+\sum\nolimits_{q \in \K_u\backslash \ell }\UIqlpzf s_q^\ul
\nonumber\\
&
+ \sum\nolimits_{q \in \K_d} \MIqlpzf s_q^\dl
+\ANlpzf,
\end{align}
where
\vspace{-0.2em}
\begin{subequations}
\begin{align}
\DSlpzf=&  \sqrt{\rho_u}
\Big(\sum\nolimits_{m \in \MM} 
\dm^{\Zlup}\alphml b_m \sqrt{ \varsigma_{\ell}}  {(\wmlulzf)^\dag}\gmlu
\nonumber\\
&\hspace{0em}
+
\dm^{\Tlup}
\alphml b_m\sqrt{ \varsigma_{\ell}}  {(\wmlulmr)^\dag}\gmlu \Big), 
\\
\UIqlpzf=&
\sqrt{\rho_u} 
\Big(\sum\nolimits_{m \in \MM}
\dm^{\Zlup}\alphml b_m \sqrt{ \varsigma_{q}} {(\wmlulzf)^\dag}\gmqu
\nonumber\\
&\hspace{0em}
+
\dm^{\Zlup}\alphml b_m \sqrt{ \varsigma_{q}} {(\wmlulmr)^\dag}\gmqu  \Big),
\\ 
\MIqlpzf  =& \sqrt{\rho_{d}} \sum\nolimits_{i \in \MM} a_i{\theta_{iq}}  
\bigg(
\sum\nolimits_{m \in \MM}
\dm^{\Zlup}
\alphml b_m 
 {\Big(\wmlulzf\Big)^\dag}
 \nonumber\\
&\hspace{0em}
\times{\qF_{mi}}\Big(\delta_i^{\Zqdl} \wiqdlzf \! + \! \delta_i^{\Tqdl} \wiqdlmr\Big)\!+\!
 \dm^{\Tlup}
 \alphml b_m 
 \nonumber\\
 &\hspace{-1.5em}
\times{\Big(\wmlulmr\Big)^\dag}{\qF_{mi}}\Big(\delta_i^{\Zqdl} \wiqdlzf  \!+\!  \delta_i^{\Tqdl} \wiqdlmr\Big)\!\bigg),
\\
\ANlpzf=&  \sum\nolimits_{m \in \MM}
\dm^{\Zlup}
\alphml b_m {(\wmlulzf)^\dag  \qw_{m}^{\ul}}
\nonumber\\
&\hspace{0em}+
\dm^{\Tlup}
\alphml b_m {(\wmlulmr)^\dag  \qw_{m}^{\ul}}.
\end{align}
\end{subequations}
Now, $s_{\ell}^{\ul}$ is detected from $r_{\ell}^{\PZF}$.

\begin{proposition}\label{Prop:SE:ULZFZF}
The achievable UL SE for the $\UEul$ at the DCCU with PZF/PZF design is given by
\begin{align}
    \label{eq:UL:SE}
	\mathcal{S}_{\ul,\ell}^{\PZF} (\qa,\qb, \boldsymbol{\varsigma}, \boldsymbol{\theta}, \boldsymbol{\alpha} )
	\!=\! \frac{\tau_c\!-\!\tau_t}{\tau_c}\log_2
	(1\!+\! \SINR_{\ul,\ell}^{\PZF}(\qa, \qb, \boldsymbol{\varsigma}, \boldsymbol{\theta}, \boldsymbol{\alpha})),
\end{align}
where
\vspace{-0.4em}
\begin{align}\label{eq:SINRMA_ZF1}
&\SINR_{\ul,\ell}^{\PZF}(\qa, \qb, \boldsymbol{\varsigma}, \boldsymbol{\theta}, \boldsymbol{\alpha})=
\nonumber\\
&
\frac{
	\rho_{u} \Big(
 \sum\nolimits_{\substack{m\in\MM}} 
 \dm^{\Zlup}\alphml b_m\sqrt{  \varsigma_{\ell}} \gamuml +\Ntx
 \dm^{\Tlup}\alphml b_m\sqrt{ \varsigma_{\ell}} {\gamuml}\Big)^2
	}
	{\Psi_{\ell}^\PZF(\qb, \boldsymbol{\varsigma},  \boldsymbol{\alpha})+
		\rho_d 
		\sum\nolimits_{\substack{i\in\MM}}
		\sum\nolimits_{q\in\K_d}
		 a_{i}\theta_{iq}^2 \gamma_{iq}^{\dl}\Phi_{\ell i}^\PZF(\qb, \boldsymbol{\alpha})
},
\end{align}
with 
\vspace{-0.4em}
\begin{subequations}~\label{eq:vhomli:mumlpzf}
\begin{align}
\Psi_{\ell}^\PZF&\triangleq
  \sum\limits_{m\in\MM}
  \Big(\rho_{u}
		 \sum\limits_{\ell'\in\K_u}\!
   \Big(
   \dm^{\Zlup}
				\alphml^2 b_m \varsigma_{\ell'}
	\frac{\gamuml(\beta_{m\ell'}^{\ul}-
		\gamumlp)}{\Ntx-|\Smu|}
  \nonumber\\
  &
 +
	\dm^{\Tlup}
     \Ntx
     \alphml^2 b_m \varsigma_{\ell'}
	{\gamuml\beta_{m\ell'}^{\ul}}\Big)
 +
		\dm^{\Zlup}
  \alphml^2b_m\frac{\gamuml}{\Ntx-|\Smu|}
 \nonumber
  \\
  &\hspace{3em}	+
  \dm^{\Tlup}\Ntx
		\alphml^2b_m\gamuml\Big),
  \\
 \Phi_{\ell i}^\PZF&\!\triangleq
 \sum\nolimits_{m\in\MM}\!\!
  b_m
\alphml^2 
\frac{\dm^{\Zlup}\gamuml}{N \!- \!\vert \Smu\vert}
     \Big(  \frac{\delta_i^{\Zqdl}\beta_{mi}^{\dl}} {N\!- \!\vert \Sm\vert }   \!+\!\delta_i^{\Tqdl}N\beta_{mi}^{\dl} \Big)
     \nonumber\\
     &\hspace{0em} 
     \!+\!\dm^{\Tlup} N
 \alphml^2 b_m
 {\gamuml}
     \Big( \frac{\delta_i^{\Zqdl}\beta_{mi}^{\dl} } {N \!- \!\vert \Sm\vert }
   +  \delta_i^{\Tqdl}N \beta_{mi}^{\dl} \Big),~\label{eq:vho:mli:pzf}
\end{align}    
\end{subequations}
where $\beta_{mm}\triangleq \SIm$.
\end{proposition}

\begin{proof}
Since the proof follows the same procedure as in the DL scenario, it is omitted here for brevity.
\end{proof}

\begin{Remark}
   The UL SE in Proposition~\ref{Prop:SE:ULZFZF} includes different combining/precoding designs in NAFD CF-mMIMO systems. More specifically, 
   \begin{itemize}
       \item PZF combining/ MR precoding:  set $\delta_i^{\Tqdl}=1$ and $\delta_i^{\Zqdl}=0$ in $\Phi_{m\ell i}^\PZF$.
       \item MR combining/ MR precoding: set $\delta_i^{\Tqdl}=1$, $\delta_i^{\Zqdl}=0$, $\dm^{\Zlup}=0$, and $\dm^{\Tlup}=1$.
       \item FZF combining/ MR precoding: set $\delta_i^{\Tqdl}=1$, $\delta_i^{\Zqdl}=0$, $\dm^{\Zlup}=1$, $\dm^{\Tlup}=0$, and $\vert \Smu\vert = K_u$. 
        \item FZF combining/ FZF precoding: set $\delta_i^{\Tqdl}=0$, $\delta_i^{\Zqdl}=1$, $\dm^{\Zlup}=1$, $\dm^{\Tlup}=0$, $\vert \Smu\vert = K_u$ and $\vert \Sm\vert = K_d$. 
   \end{itemize}
   \label{remark_2}
\end{Remark}

\vspace{-1em}
\subsection{Sum Spectral Efficiency Maximization}
\label{sec:SE:opt}
We now formulate an optimization problem to maximize the sum SE of the considered network by optimizing the mode assignment vectors $(\aaa,\bb)$, the UL LSFD weights 
$\ALPHA$, and the UL and DL power control coefficients 
$(\THeta, \VARSIGMA)$. The optimization constraints include the per-UE SE requirement, as well as the maximum available transmit power at each AP and UL UE. More specifically, the optimization problem is
\begin{subequations}\label{P:SE}
	\begin{align}
		\underset{\qx}{\max}\,\, &
		\SSS^{\xi,\vFD} (\x)
		\\
		\mathrm{s.t.} \,\,
		\nonumber
		& \eqref{a}-\eqref{sumabFD}, 
		\eqref{DL:power:cons2}, \eqref{UL:power:cons}, 
  \eqref{eq:alphml}, 
		\\
		& \mathcal{S}_{\ul,\ell}^{\xi,\vFD} (\qb, \boldsymbol \varsigma, {\boldsymbol{\theta}},\boldsymbol \alpha) \geq \mathcal{S}_\ul^o,~\forall \ell
		\label{UL:QoS:cons}
		\\
		&\mathcal{S}_{\dl,k}^{\xi,\vFD} (\qa, \boldsymbol \theta, {\boldsymbol{\varsigma}}) \geq  \mathcal{S}_\dl^o,~\forall k, 
		\label{DL:QoS:cons}
  \\
  & \theta_{mk} \geq 0,~\forall k, m,
	\end{align}
\end{subequations}
where $\xi\in\{\PZF,\MR, \FZF\}$, $\qx\triangleq\{\qa, \qb, \boldsymbol \varsigma, \boldsymbol \theta, \boldsymbol \alpha\}$, $\mathcal{S}_\ul^o$ and $\mathcal{S}_\dl^o$ are the minimum SE requirements of $\UEul$ and $\UEdk$, respectively, and the total SE of the NAFD system is
\begin{align}
    \SSS ^{\xi,\vFD}(\x) &\triangleq \sum\nolimits_{\ell\in\mathcal{K}_u} \mathcal{S}_{\ul,\ell}^{\xi,\vFD} (\qb, \boldsymbol \varsigma, \boldsymbol \theta, \boldsymbol \alpha)   \nonumber\\
    &+  \sum\nolimits_{k\in\mathcal{K}_d}\mathcal{S}_{\dl,k} ^{\xi,\vFD}(\qa, \boldsymbol \theta, {\boldsymbol{\varsigma}}).
\end{align} 
We emphasize that problem~\eqref{P:SE} optimizes a set of different variables to simultaneously satisfy the quality of services to the $K_u +K_d$ UEs.  A full service satisfaction to all the UEs might be impossible in challenging circumstances such as harsh propagation conditions with limited power budget in both the UL and DL data transmissions. Note that if only one user does not meet their Se requirement, then the feasible region of the  problem~\eqref{P:SE} is empty and there is no solution found. The innovative idea is to detect which users should be scheduled or which subset of UEs should be served. In our framework, this selection process is not arbitrary but coupled with  problem~\eqref{P:SE}. Specifically, the feasibility of the optimization problem directly depends on whether the chosen subset of UEs can be simultaneously supported under the stringent SE and power constraints. We emphasize that the achievable UL and DL SE of each user is obtained based on the quasi-static channel model. Consequently, the feasibility of  problem~\eqref{P:SE} is over multiple coherence blocks whenever the large-scale fading coefficients remain.
\section{Proposed Differential Evolution Algorithm}
\label{sec:solution}
We propose a constraint-handling differential evolution (CHDE) algorithm, which combines a SE lower boundary constraint-handling mechanism with DE. The overall flow of the proposed algorithm is illustrated in Algorithm \ref{alg:CHDE}. CHDE begins with a randomly initialized population including $I$ individuals, while each individual corresponds to a solution. In each generation, all individuals participated in the evolutionary process through two operators: mutation and crossover. Mutation introduces diversity by perturbing existing solutions, while crossover combines current solutions and mutant vectors to generate candidate offspring. Additionally, CHDE retains only top solutions across generations and handles constraints by stopping service to UEs below the SE threshold and reallocating power to the remaining UEs, thereby improving efficiency. Upon termination, the best individual is returned as the solution.
\vspace{-1em}
\subsection{Solution Representation and Population Initialization}
Each individual represents a potential solution, and the fitness value of an individual corresponds to the objective of the solution to  problem \eqref{P:SE}. At the beginning of the algorithm, a population $\mathcal{P}^{(G)}$ 
consisting of $I$ individuals. Each individual is represented by an $len$-dimensional vector, where $len = 2M + K_u + MK_u + MK_d + 1$ and each element is in $[0,1]$. Specifically, the $i$-th individual of the population at $G$-th generation is expressed as follows
\begin{equation}
    \mathbf{x}^{(iG)} = [x^{(iG)}_1,x^{(iG)}_2, \ldots, x^{(iG)}_{len}]^T,
\end{equation}
Since this representation is continuous, a decoding step is required to map the vector elements into decision variables of problem \eqref{P:SE}. In particular, the first $2M$ elements correspond to the operation modes of APs, denoted by $\big\{x^{(iG)}_{n}\big\}_{n = 1}^{M}$ and $\{x^{(iG)}_{n}\}_{n = M + 1}^{2M}$. Specifically, $\forall m \in [1, \ldots, M]$, we have
\begin{equation}
    a_m = \mathbf{1}\!\left(x^{(iG)}_m \geq 0.5\right), 
\quad 
b_m = \mathbf{1}\!\left(x^{(iG)}_{m+M} \geq 0.5\right),
\end{equation}
where $\mathbf{1}(\cdot)$ denotes the indicator function. Additionally, the segments $\{x^{(iG)}_{n}\}_{n = 2M + 1}^{2M + K_u}$ and $\{x^{(iG)}_{n}\}_{n = 2M + K_u + 1}^{2M + K_u + MK_u}$ are adaptive to be converted into $\boldsymbol \varsigma$ and $\boldsymbol \alpha$, respectively. Finally, $\{x^{(iG)}_{n}\}_{n = 2M + K_u + MK_u + 1}^{2M + K_u + MK_u + MK_d}$ and $x^{(iG)}_{2M + K_u + MK_u + MK_d +1}$ represent how the power allocation satisfies constraint (\ref{DL:power:cons2}), specifically as follows:
\begin{equation}
    \theta_{mk} = \sqrt{\frac{a_m x^{(iG)}_{len}}{\gamdmk \varphi_{mk}}\frac{x^{(iG)}_{2M + K_u + MK_u +(m-1)K_d + k}}{\sum_{n =2M + K_u + MK_u + (m-1)K_d + 1 }^{2M + K_u + MK_u + mK_d}x^{(iG)}_n}}.
\end{equation}
Otherwise, we denote $\mathbf{x}^{(bestG)}$ as the individual with the maximum fitness in the $G$-th generation obtained by evaluating the objective function of problem (\ref{P:SE}) for the current population.

\begin{algorithm}[t]
    \caption{CHDE}
    \label{alg:CHDE}
        \textbf{Input:} NAFD CF-mMIMO; scale factor $\mathsf{F}$, crossover rate $\mathsf{CR}$, and maximum number of generations $G_{\max}$;\\
        \textbf{Output:}  $\qa, \qb, \boldsymbol \varsigma, \boldsymbol \theta,~\text{and}~ \boldsymbol\alpha$; 
    \begin{algorithmic}[1] 
        \STATE Generation index $G \leftarrow 1$;
        \STATE Randomly initialize a population $\mathcal{P}^{(G)}$ of $I$ individuals;
        \STATE Repair and evaluate each individual within $\mathcal{P}^{(G)}$;
        \WHILE{\textit{Termination condition not met}}
            \STATE Offspring  $\mathcal{B}^{(G)} \leftarrow \varnothing$; 
            \FOR{$i\leftarrow 0$ to $I$}
                \STATE Generate mutant vector $\mathbf{u}^{(iG)}$ using \eqref{eq:mutation} and \eqref{eq:normalize_mutant};
                \STATE Generate trial solution $\mathbf{v}^{(iG)}$  
                using \eqref{eq:crossover};
                \STATE Repair and evaluate $\mathbf{v}^{(iG)}$;
                \IF{$f\left(\mathbf{v}^{(iG)}\right) \geq f\left(\mathbf{x}^{(iG)}\right)$}
                    \STATE $\mathcal{B}^{(G)} \leftarrow \mathcal{B}^{(G)} \cup \mathbf{v}^{(iG)} $;
                \ELSE
                    \STATE $\mathcal{B}^{(G)} \leftarrow \mathcal{B}^{(G)} \cup \mathbf{x}^{(iG)} $
                \ENDIF
            \ENDFOR
            \STATE $\mathcal{P}^{(G+1)} \leftarrow \mathcal{B}^{(G)}$; $G \leftarrow G + 1$;
        \ENDWHILE
    \end{algorithmic}
\end{algorithm}
\setlength{\textfloatsep}{0.05cm}
\vspace{-1em}
\subsection{Individual Repairing and Evaluation with Infeasibility}
The quality of each individual within the population is evaluated across the optimal capability of the optimization problem. The fitness of each individual $\mathbf{x}^{(iG)}$ is defined from the objective function 
as follows:
\begin{equation}
    f\big(\mathbf{x}^{(iG)}\big) = \SSS^{\xi,\vFD} \big(\x^{(iG)}\big).
    \label{eq:fitness_cal}
\end{equation}

Corresponding to the maximization problem, an individual $\mathbf{x}^{(iG)}$ is better than another $\mathbf{x}^{(jG)}$ when $f(\mathbf{x}^{(iG)})$ is larger than $f(\mathbf{x}^{(jG)})$. However, to address constraints (\ref{UL:QoS:cons}) and (\ref{DL:QoS:cons}), we employ a technique to suspend service for idle UEs without ensuring power constraints. Specifically, assuming $\tilde{\mathcal{K}}_u$ and $\tilde{\mathcal{K}}_d$ are sets of UL and DL UEs, respectively, who do not meet the lower power boundary, defined as follows:
\begin{align}
    \tilde{\mathcal{K}}_u &= \{\ell: \mathcal{S}_{\ul,\ell}^{\xi,\vFD} (\qb, \boldsymbol \varsigma, {\boldsymbol{\theta}},\boldsymbol \alpha) < \mathcal{S}_\ul^o,~\forall \ell \in \mathcal{K}_u \},\\
    \tilde{\mathcal{K}_d} &= \{ k: \mathcal{S}_{\dl,k}^{\xi,\vFD} (\qa, \boldsymbol \theta, {\boldsymbol{\varsigma}}) < \mathcal{S}_\dl^o,~\forall k \in \mathcal{K}_d\}.
\end{align}

Denote $\mathbf{U}^c$ and $\mathbf{D}^c$ as two diagonal matrices representing the corresponding service  for UL and DL UEs, defined as:
\begin{equation}
    U^c_{\ell} = \mathbf{1}\!\left(\ell \in \mathcal{K}_u \setminus \tilde{\mathcal{K}}_u\right), 
\quad 
D^c_{k} = \mathbf{1}\!\left(k \in \mathcal{K}_d \setminus \tilde{\mathcal{K}}_d\right),
\end{equation}
where $\mathbf{1}(\cdot)$ denotes the indicator function. Then, the solution related to the individual $\mathbf{x}^{(iG)}$ is updated:
\begin{align}
    \tilde{\boldsymbol \varsigma} = \mathbf{U}^c \boldsymbol \varsigma, \quad \tilde{\boldsymbol \alpha} = \mathbf{D}^c \boldsymbol \alpha,
\end{align}
\begin{equation}
    \tilde{\theta}_{mk} = \sqrt{\frac{a_m x^{(iG)}_{len}}{\gamdmk \varphi_{mk}}\frac{x^{(iG)}_{2M + K_u + MK_u +(m-1)K_d + k}D^c_k}{\sum_{n =2M + K_u + MK_u + (m-1)K_d + 1 }^{2M + K_u + MK_u + mK_d}x^{(iG)}_nD^c_k}},
\end{equation}
where, $\tilde{\boldsymbol \varsigma}$, $\tilde{\boldsymbol \alpha}$ and $\tilde{\theta}$ have some zero elements at positions of UEs that are not served. Therefore, the SE of each UE and the fitness of individual $\mathbf{x}^{(iG)}$ are calculated again as follows:
\begin{align}
    \tilde{\mathcal{S}}_{\ul,\ell}^{\xi,\vFD} &= \begin{cases}
        \mathcal{S}_{\ul,\ell}^{\xi,\vFD}(\qb, \tilde{\boldsymbol \varsigma}, \tilde{{\boldsymbol{\theta}}},\tilde{\boldsymbol \alpha}), ~\forall \ell \in \mathcal{K}_u \backslash \tilde{\mathcal{K}}_u,\\
        0, ~\forall \ell \in \tilde{\mathcal{K}}_u,
    \end{cases}
\\
    \tilde{\mathcal{S}}_{\dl,k}^{\xi,\vFD} &= \begin{cases}
        \mathcal{S}_{\dl,k}^{\xi,\vFD}(\qa, \tilde{{\boldsymbol{\theta}}}, \tilde{\boldsymbol \varsigma}), ~\forall k \in \mathcal{K}_d \backslash \tilde{\mathcal{K}}_d,\\
        0, ~\forall k \in \tilde{\mathcal{K}}_d,
    \end{cases}
\end{align}
\begin{equation}
    f(\mathbf{x}^{(iG)}) = \sum\nolimits_{\ell \in \mathcal{K}_u} \tilde{\mathcal{S}}_{\ul,\ell}^{\xi,\vFD} + \sum\nolimits_{k \in \mathcal{K}_d}\tilde{\mathcal{S}}_{\dl,k}^{\xi,\vFD}.
    \label{eq:update_fitness}
\end{equation}

The suspend-service mechanism is applied per large-scale fading timescale. Dropped UEs are reinserted and re-evaluated in the nsext time slot; therefore, our method does not permanently starve a UE unless it remains infeasible over consecutive timescales due to the power budget\footnote{Such persistently infeasible UEs can be handled by higher-layer scheduling mechanisms, e.g., QoS constraint relaxing, while the proposed algorithm focuses on providing a feasible physical-layer configuration.}.
\vspace{-1em}
\subsection{Reproduction and individual selection}
In the $G-$th generation, each individual $\mathbf{x}^{(iG)}$ will create an offspring individual across mutation and crossover operators, aiming to explore and exploit the solution space. In this paper, we first use the well-established \textit{DE/current-to-$p$best/1} operator to create a mutant vector, which has been widely used in advanced DE variants such as JADE, LSHADE, and their various extensions \cite{LSHADE_wang2021shade, LSHADE_tang2022adaptive} due to the robust balance between convergence speed and population diversity. The mutant vector is mathematically described as
\begin{equation}
    \mathbf{u}^{(iG)} = \mathbf{x}^{(iG)} + \mathsf{F}\left(\mathbf{x}^{(pbestG)} - \mathbf{x}^{(iG)}\right) +  \mathsf{F}\left(\mathbf{x}^{(r_1G)} - \mathbf{x}^{(r_2G)}\right),
    \label{eq:mutation}
\end{equation}
where $\mathbf{u}^{(iG)}$ is a mutant vector corresponding to the parent solution $\mathbf{x}^{(iG)}$, $\mathbf{x}^{(pbestG)}$ is selected randomly from the top best solutions, $\mathbf{x}^{(r_1G)}$ and $\mathbf{x}^{(r_2G)}$ are selected randomly from current population, and $\mathsf{F}$ is a scaled factor. Specifically, $\mathbf{x}^{(pbestG)} - \mathbf{x}^{(iG)}$ represents a local search behavior aimed at guiding the individual's movement to exploit the promising solution space provided by the best individual. The second component $\mathbf{x}^{(r_1G)} - \mathbf{x}^{(r_2G)}$ represents the disparity between the two randomly selected individuals, serving as a global search behavior.  Nonetheless, \eqref{eq:mutation} may generate some elements of $\mathbf{u}^{(iG)}$ are outside the search range boundaries $[x_j^{min}, x_j^{max}]$, we applied the correction performed in \cite{lshade}:
\begin{equation}
    u_j^{(iG)} = \begin{cases}
        (x^{min}_j + x_j^{(iG)})/2\quad \text{if} ~u_j^{(iG)} < x^{min}_j,\\
        (x^{max}_j + x_j^{(iG)})/2\quad \text{if} ~u_j^{(iG)} > x^{max}_j.
    \end{cases}
    \label{eq:normalize_mutant}
\end{equation}

After generating the mutant vector, $\mathbf{u}^{(iG)}$ is crossed with the parent $\mathbf{x}^{(iG)}$ to generate trial vector $\mathbf{v}^{(iG)}$. In this paper, Binomial Crossover, which is the most commonly used crossover operator in DE, is employed and implemented as
\begin{equation}
    v_j^{(iG)} = \begin{cases}
        u_j^{(iG)} \quad 
        &\mathcal{U}[0,1) \leq \mathsf{CR}~\text{or}~j=j_{rand},\\
        x_j^{(iG)} \quad &\text{otherwise},
    \end{cases}
    \label{eq:crossover}
\end{equation}
where $\mathcal{U}[0, 1)$ denotes a uniformly random number in range $[0,1)$, and  $j_{rand} \in \{1, 2, \ldots, len\}$ is randomly preselected to ensure that $\mathbf{v}^{(iG)}$ has at least on dimension different from $\mathbf{x}^{(iG)}$ aiming to create the diversity within the population. After generating all trial vectors $\mathbf{v}^{(iG)}$, their fitness is evaluated by (\ref{eq:update_fitness}). Selection process then compares each $\mathbf{x}^{(iG)}$ with its trial vector $\mathbf{v}^{(iG)}$, keeping the better one for the next generation.
\begin{equation}
    \mathbf{x}^{(i(G+1))} = \begin{cases}
        \mathbf{v}^{(iG)} &\text{if}~ f(\mathbf{v}^{(iG)}) \geq f(\mathbf{x}^{(iG)}), \\
        \mathbf{x}^{(iG)} &\text{otherwise}.
    \end{cases}
\end{equation}

\vspace{-1em}
\subsection{Termination condition}
CHDE will be terminated if one of the following two criteria is satisfied: \textit{i)} the number of generations reaches a maximum specified, thereby ensuring bounded computational complexity and preventing excessive runtime and \textit{ii)} the objective function remains unchanged over predetermined generations, which implies the search process convergence to the stable solution. These criteria jointly provide a balance between solution quality and computational efficiency, allowing the algorithm to terminate either when sufficient exploration has been completed or when further iterations are unlikely to yield improvement.

\subsection{Discussion of proposed algorithm}
We  probabilistically analyze the convergence of CHDE to an optimal solution using an $\mathcal{\varepsilon}$-optimal solution space. 
\begin{Definition}
    Let us introduce $\mathcal{S}^{*}_{\varepsilon}$ to be the space of the $\varepsilon-$optimal solution to (\ref{P:SE}), which is
    \begin{equation}
        \mathcal{S}^{*}_{\varepsilon} = \left\{\mathbf{x} \big||f(\mathbf{x}) - f(\mathbf{x}^{*})| \leq \varepsilon, \mathbf{x}~and~\mathbf{x}^{*} \in \mathcal{S} \right\},
        \label{eq:epsilong_space}
    \end{equation}
    where $f(\cdot)$ and $\mathcal{S}$ denote the objective function and the feasible space of  problem (\ref{P:SE}), respectively. However, $\mathbf{x}^{*}$ is the optimal solution and $\varepsilon$ is a small positive value.
\end{Definition}

The set $\mathcal{S}^{*}_{\varepsilon}$  consists of all solutions whose objective values deviate from the global optimum by at most $\varepsilon$. From an initial point, \textbf{Algorithm~\ref{alg:CHDE}} gradually improves population quality along generations, therefore the population tends to move closer to regions of higher fitness, which increases the likelihood that at least one individual lies within an $\mathcal{\varepsilon}-$optimal solutions space.
\begin{Lemma}\cite{chien2024active}
\label{lemma:convergence_probability}
    For a population $\mathcal{P}$ with $I$ individuals, the probability of the population converging to an individual belonging to $\mathcal{S}_{\varepsilon}^{*}$ by exploiting CHDE is defined as follows
\begin{equation}
    \label{eq:convergence}
    \mathsf{Pr}\left(\mathcal{P} \cap \mathcal{S}^{\ast}_{\varepsilon} \neq \varnothing\right) \geq  1- \left(1 - \mu (\mathcal{S}^{\ast}) P_{\mathrm{ep}}\right)^{I},
\end{equation}
where $P_{ep}$ is the mutation probability of each individual, $\mu(\mathcal{S}_\varepsilon^\ast)$ is the measure to $\mathcal{S}_\varepsilon^\ast$ regarding the reproduction.
\end{Lemma}

\begin{Lemma}
    \label{lemma:convergence_to_optimal}\cite{chien2024active}  Let $\mathcal{P}^{(G)}$ be the population at the $G-th$ generation. Then, the convergence to the global optimum by CHDE is
        \begin{equation}
    \lim_{G\rightarrow \infty}\  \mathsf{Pr}\left(\mathcal{P}^{(G)}\cap \mathcal{S}^*_{\varepsilon} \neq \varnothing \right) = 1.
\end{equation}
\end{Lemma}

Lemmas~\ref{lemma:convergence_probability} and \ref{lemma:convergence_to_optimal} indicate that CHDE converges to the global optimum from an initial population within the feasible domain after a sufficiently large number of generations. Therefore, we continuously analyze the ideal number of generations for CHDE to converge toward the space $\mathcal{S}_{\varepsilon}^{*}$.

\begin{Definition}
    Assume that $\mathbf{x}^{*}_{\varepsilon}$ is a solution in the solution space $\mathcal{S}^{*}_{\varepsilon}$, let $f_{best}^{(G)} = f(\mathbf{x}^{(bestG)})$ be the best fitness of population $\mathcal{P}^{(G)}$, define $\Delta^{(G)}_{\varepsilon} = \max \{f(\mathbf{x}^{*}_{\varepsilon}) - f_{best}^{(G)}, 0\}$, which is used to measure the distance of the population to the $\varepsilon-$optimal solution. Obviously, the sequence $\{\Delta^{(G)}_{\varepsilon}\}_{G = 1}^{\infty}$ generated by CHDE is a non-negative stochastic process. Then, the stopping time of CHDE is $T_{\varepsilon} = \min\{G \geq 1 : \Delta^{(G)}_{\varepsilon} = 0\}$, which is the first hitting time \cite{first_hitting_time_yu2008new} for the algorithm towards an $\varepsilon-$optimal solutions for problem (\ref{P:SE}). 
\end{Definition}

\begin{Definition}
    Assume that $\{\Delta^{(G)}_{\varepsilon}\}_{G = 1}^{\infty}$ is a stochastic process on a probability space $(\Omega, \mathcal{F}, \mathsf{Pr})$. The natural filtration of $\mathcal{F}$ is denoted by $\mathcal{F}^{(G)} = \sigma(\Delta^{(1)}_{\varepsilon}, \Delta^{(2)}_{\varepsilon}, \ldots, \Delta^{(G)}_{\varepsilon})$. At the generation $G$, the average gain is expressed as follows
    \begin{equation}
        g^{(G)} = \mathbb{E}\left\{\Delta^{(G)}_{\varepsilon} - \Delta^{(G+1)}_{\varepsilon} \big| \mathcal{F}^{(G)}\right\}.
        \label{eq:gain_average}
    \end{equation}  
\end{Definition}

\begin{figure*}
\centering
\begin{subfigure}{0.32\textwidth}
    \includegraphics[width=\linewidth]{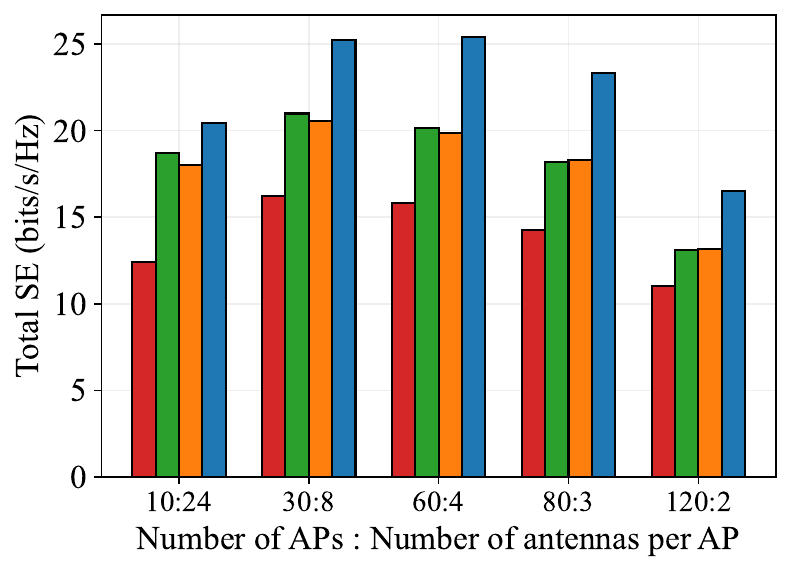}
    \vspace{-0.5em}
    \caption{$K_u = 3, K_d = 3$}
    \label{fig:first_objective_optimization}
\end{subfigure}
\hfill
\begin{subfigure}{0.32\textwidth}
    \includegraphics[width=\textwidth]{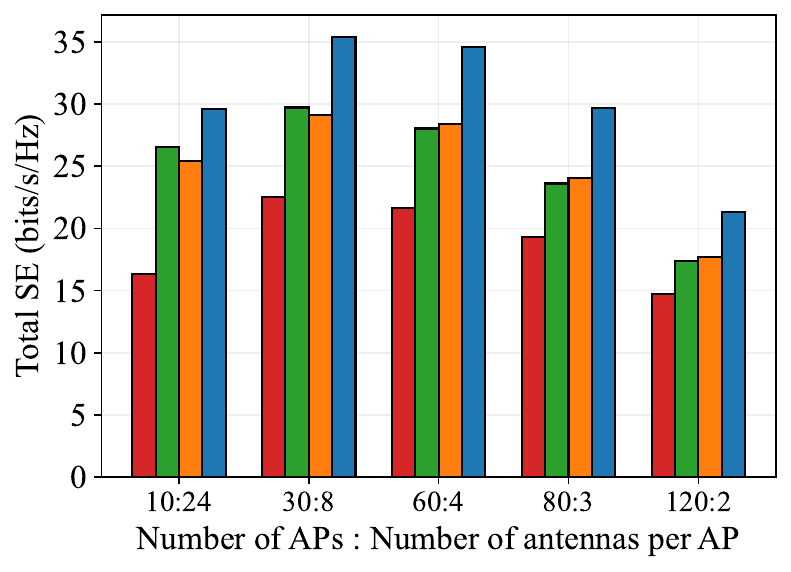}
    \vspace{-0.5em}
    \caption{$K_u = 5, K_d = 5$}
    \label{fig:second_objective_optimization}
\end{subfigure}
\begin{subfigure}{0.32\textwidth}
    \includegraphics[width=\textwidth]{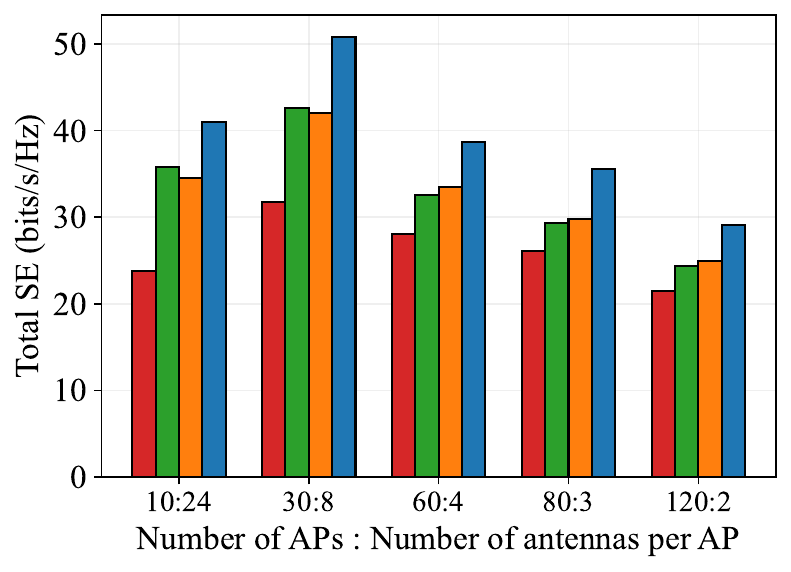}
    \vspace{-0.5em}
    \caption{$K_u = 10, K_d = 10$}
    \label{fig:third_objective_optimization}
\end{subfigure}
\begin{subfigure}{1\textwidth}
    \centering
    \includegraphics[width=0.6\textwidth, trim = 1.5cm 26.7cm 25.5cm 2.5cm, clip = true]{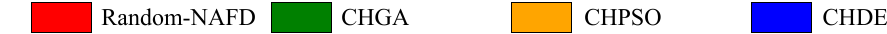}
\end{subfigure}        
\caption{Total SE of different benchmarks comprising CHDE, CHGA, CHPSO, and Random-NAFD.}
\label{fig:objective_otpimization}
\vspace{-1.2em}
\end{figure*}
In \cite{first_hitting_time_yu2008new}, evolutionary process $\{\mathcal{P}^{(G)}\}_{G = 1}^{\infty}$ of CHDE is modeled by a Markov chain. Accordingly, the behavior of population $\mathcal{P}^{(G+1)}$ 
depends only on $\mathcal{P}^{(G)}$. Therefore, the sequence $\{\Delta^{(G)}_{\varepsilon}\}_{G = 1}^{\infty}$ is also referenced to as a Markov chain, and the gain in  (\ref{eq:gain_average}) can be simplified to $g^{(G)} = \mathbb{E}\left\{\Delta^{(G)}_{\varepsilon} - \Delta^{(G + 1)}_{\varepsilon}\big| \Delta^{(G)}_{\varepsilon}\right \}$. 
Theorem \ref{theo:first_fitting_time} bounds the expected first hitting time of CHDE via average gain.
\begin{Theorem} \label{theo:first_fitting_time}
Let $\{\Delta^{(G)}_{\varepsilon}\}_{G = 1}^{\infty}$ be a stochastic process associated with CHDE, where $\Delta^{(G)}_{\varepsilon} \geq 0, \forall G \geq 1$. Suppose $h: (0, \Delta^{(1)}_{\varepsilon}] \rightarrow \mathbb{R}^{+}$ be a monotonically increasing continuous function. If $\mathbb{E}\left\{\Delta^{(G)}_{\varepsilon} - \Delta^{(G + 1)}_{\varepsilon} \big| \Delta^{(G)}_{\varepsilon}\right\} \geq h(\Delta^{(G)}_{\varepsilon})$ , when $\Delta^{(G)}_{\varepsilon} > 0$, then
  $\mathbb{E} \left\{T_\varepsilon \big| \Delta^{(1)}_{\varepsilon}\right\} \leq 1 + \int_{0}^{\Delta^{(1)}_{\varepsilon}} \frac{1}{h(t)}dt.$
    
\end{Theorem}
\begin{proof}
    See Appendix~\ref{proof:first_hitting_time}.
\end{proof}
Theorem~\ref{theo:first_fitting_time} provides a theoretical bound on convergence speed, linking the expected stopping time of CHDE to the magnitude of improvements obtained per generation.

\textit{\textbf{Complexity Analysis:}} We now analyze the computational complexity of the proposed algorithm based on basic operations with the complexity of $\mathcal{O}(1)$, including addition, subtraction, multiplication, and inversion. Based on Remark \ref{remark_1} and Remark \ref{remark_2}, the computational complexity of the fitness function $f(\cdot)$, which is defined in \eqref{eq:fitness_cal}, is independent of the linear proceeding technique $\xi$. Without loss of generality, we analyze the computational complexity of the fitness function with PZF precoding. For each $\UEdk$, evaluating the SE $\mathcal{S}_{\dl,k}^{\PZF}$ requires a computational complexity in the order of $\mathcal{O}(MK_d)$. The complexity to compute the SE for $\UEul$ is in the order of $\mathcal{O}(MK_u + M^2K_d)$. Therefore, the complexity to compute the fitness function is $\mathcal{O}(MK_d^2 + MK_u^2 + M^2K_dK_u)$. Each individual can be evaluated up to a maximum of ($K_d + K_u$) times to ensure constraints \eqref{UL:QoS:cons} and \eqref{DL:QoS:cons} (the maximum of evaluations occurs when there is still one UE who has not yet met the lower SE requirement after each evaluation). The complexity to evaluate each individual is $\mathcal{O}((K_d + K_u)(MK_d^2 + MK_u^2 + M^2K_dK_u))$. Let $FEO = (K_d + K_u)(MK_d^2 + MK_u^2 + M^2K_dK_u)$. Regarding the computational complexity of CHDE, the population initialization step requires $\mathcal{O}(I\times len + I\times FEO)$, where $len = 2M+K_u + MK_u + MK_d +1$. In each generation, the population sorting requires $\mathcal{O}(I\log(I)$ to select the top best solution. The reproduction step requires $\mathcal{O}(I\times len)$. The individual repairing and evaluation processing within the offspring requires  $\mathcal{O}(I\times FEO)$. The selection step requires $\mathcal{O}(I)$. Therefore, the computing complexity of the proposed algorithm is $\mathcal{O}(GI(K_u^3 + K_d^3 + MK_u^2K_d + MK_uK_d^2)$.

\begin{remark}
It is worth noting that the dominant computational cost of the proposed CHDE algorithm lies in the individual repairing and fitness evaluation stages, whose complexity scales as $\mathcal{O}(I \times FEO)$ per generation. Importantly, individuals in the population are evaluated independently, making the fitness evaluation inherently parallelizable across individuals. Therefore, although the computational complexity is analyzed under a sequential execution model, the actual wall-clock runtime can be substantially reduced by exploiting parallel computing architectures, such as multi-core Central Processing Units or Graphics Processing Units. This parallelism enables efficient implementation of CHDE in practical systems and helps mitigate the computational delay associated with large-scale network optimization.
\end{remark}

\begin{remark}
The DE-based algorithm is selected to solve  problem~\eqref{P:SE} and address the infeasible circumstances since \eqref{P:SE}  is non-convex, mixed-integer, and highly nonlinear due to the coupling between continuous power allocation and binary AP mode selection variables. Algorithm~\ref{alg:CHDE} is well-suited for such problems because it does not rely on convexity assumptions, no gradient evaluations of the objective functions and constraints, and demonstrates strong performance in exploring complex search spaces. Compared to the other metaheuristics, e.g., genetic algorithm and particle swarm optimization, the DE-based algorithm has been empirically shown to offer faster convergence and higher solution quality for continuous optimization problems.
\end{remark}

\begin{figure*}
\centering
\begin{subfigure}{0.31\textwidth}
    \includegraphics[width=\linewidth]{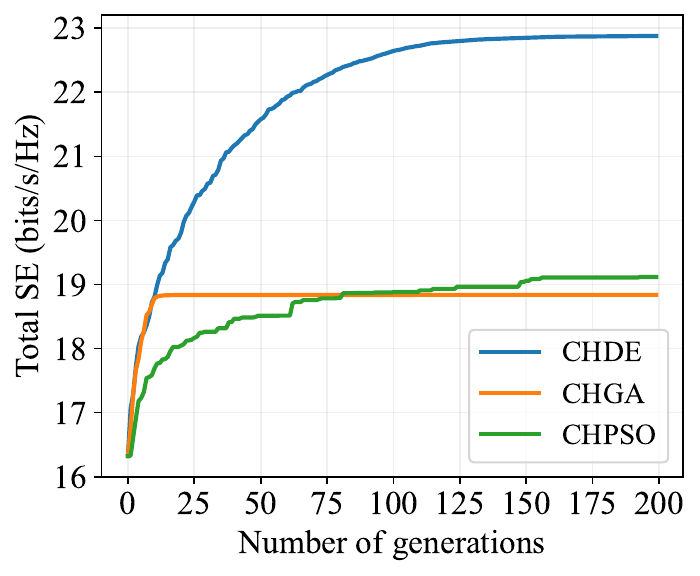}
    \caption{$K_u = 3, K_d = 3$}
    \label{fig:first_convergence}
\end{subfigure}
\begin{subfigure}{0.31\textwidth}
    \includegraphics[width=\textwidth]{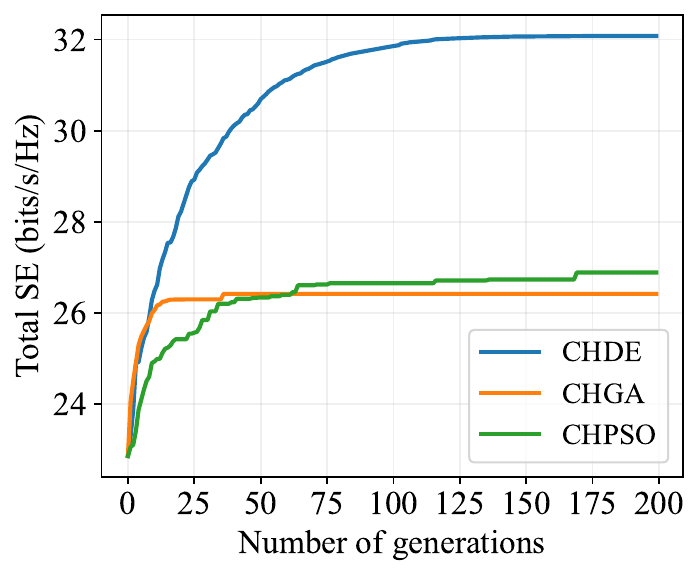}
    \caption{$K_u = 5, K_d = 5$}
    \label{fig:second_convergence}
\end{subfigure}
\begin{subfigure}{0.31\textwidth}
    \includegraphics[width=\textwidth]{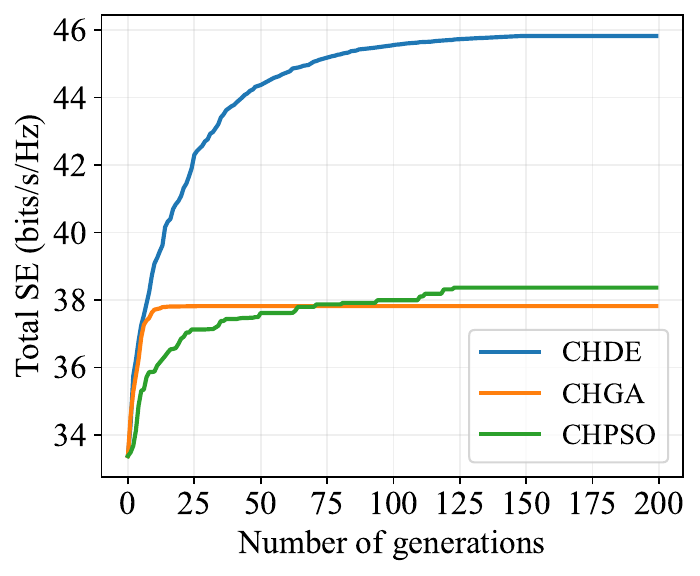}
    \caption{$K_u = 10, K_d = 10$}
    \label{fig:third_convergence}
\end{subfigure}
\caption{Convergence of different benchmarks comprising CHDE, CHGA, and CHPSO.}
\label{fig:Convergence}
 \vspace{-0.1em}
\end{figure*}

\begin{figure}[t]
    \centering
\includegraphics[width=0.75\linewidth]{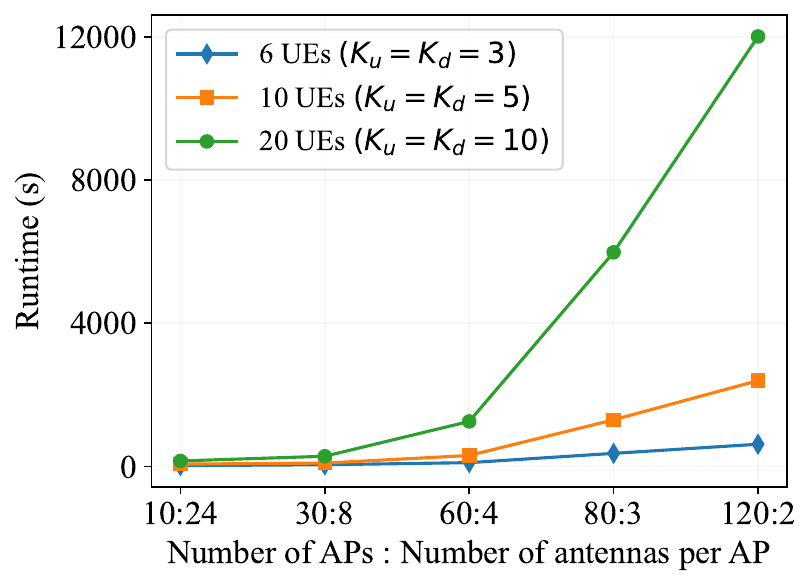}
    \caption{Runtime of CHDE.}
    \label{fig:run_time}
\end{figure}


\section{Numerical Results}~\label{Sec:Numer}
We consider an NAFD CF-mMIMO system in which the APs and UEs are dispersed randomly within a $0.5 \times 0.5$ $\text{km}^2$ square, with the edges wrapped to prevent the boundary effects. Unless otherwise stated, the values of the network parameters are: number of APs $M \in [10,  30,  60, 80, 120]$, number of transmit radio frequency chains $N \in [24,  8,   4, 3, 2]$, number of UL UEs $K_u \in [3, 5, 10]$, number of DL UEs $K_d \in [3, 5, 10]$. The number of AP antennas must be significantly larger than the number of users to avail of the benefits of massive MIMO. In the considered scenarios, this ratio spans from $12$ to $40$. The other system parameters include communication bandwidth is $B = 50$ MHz, length of the coherence block $\tau_c = 200$, length of the pilot $\tau_t = K_u + K_d$, and $\mathcal{S}_{dl}^{o} = \mathcal{S}_{ul}^{o} = \mathcal{S}_{QoS}$. Consequently, the noise power $\Sn = k_BT_0BF$, where Boltzmann constant $k_B = 1.381 \times 10^{-23}~\text{Joules}\backslash^o\text{K}$, temperature noise $T_0 = 290^oK$ and noise figure $F = 9dB$. Let $\tilde{\rho}_d = 0.8$ W, $\tilde{\rho}_u = 0.2$ W and $\tilde{\rho}_t = 0.2$ W be the maximum transmit power of the APs, UL UEs and UL training pilot sequences, respectively. The normalized maximum transmit powers $\rho_d$, $\rho_u$, and $\rho_t$ are obtained by dividing the respective transmit powers by the noise power. We set $\SIm/\sigma_n^2 = 50$ dB. The large-scale fading coefficient $\beta_{mk}$ $(\betamkd, \betamlu, \betakldu)$ is modeled by \cite{emil20TWC} as 
    $\beta_{mk} = 10^{\frac{PL_{mk}^d}{10}}10^{\frac{U_{mk}}{10}}$,
where $10^{\frac{PL_{mk}^d}{10}}$ and  $10^{\frac{U_{mk}}{10}}$ represent the path loss and represents the shadowing effect, respectively. Here, $U_{mk}$ and $PL_{mk}^d$ (in dB) are calculated as 
   $ U_{mk} \!\sim\! \mathcal{N}(0, 4^2),~PL_{mk}^d \!=\! -30.5 \!-\! 36.7\text{log}_{10}\left(d_{mk} /1 \text{m} \right)$,
where $d_{mk}$ is the distance between AP $m$ and UE (AP) $k$, calculated as the minimum distance across various wrap-around scenarios, considering the $10$-meter height difference. The shadowing terms from an AP $m,~\forall m \in \mathcal{M}$ to different UE $k \in \mathcal{K}_d$ $(\ell \in \mathcal{K}_u)$ are modeled as~\cite[Eq. (93)]{Mohammad:JSAC:2023}.

We compare our proposed CHDE algorithm with the baseline scheme, denoted as Random-NAFD,\footnote{For  Random-NAFD, APs are randomly assigned to FD/HD mode. For HD APs, UL/DL operation modes are randomly chosen. Power is equally allocated to UEs. Power control is equally designed for DL APs with full power. There is no LSFD ($\alphml = 1, \forall m, \forall \ell$).} and robust benchmark algorithms for optimization problems that incorporate constraint-handling techniques, including the genetic algorithm \cite{peng2021analysis} (CHGA) and particle swarm optimization \cite{nama2023boosting} (CHPSO). Specifically, we evaluate the performance gain of the proposed model compared to the previous model \cite{Mohammad:JSAC:2023} and the efficiency of CHDE in solving problem \eqref{P:SE}. Furthermore, we analyze the average SE per UE and the number of UEs served by the network when the minimum SE threshold increases. All the numerical results were generated using Python simulation on a personal computer with Intel\textsuperscript{\textregistered} Core i5-13500H CPU and 16 GB RAM. These simulations are primarily intended to provide theoretical insights into the algorithm's behavior and performance trends, rather than to present real-time operation.\footnote{Due to hardware limitations, the system performance and the efficiency of the considered benchmarks are evaluated offline or for systems with low mobility. For practical applications, our algorithm requires a computing system capable of solving the problem~\eqref{P:SE} in sub-milliseconds in order to respond to the rapid change of propagation environments. In more detail,  implementations using compiled languages like C++ or/hardware accelerators (e.g., GPUs, FPGAs, or ASICs) can reduce runtimes by orders of magnitude \cite{mehlhose2022gpu}.}

\begin{table*}[!http]
\caption{Impact of the DE operators on the performance of CHDE}
\vspace{-0.5em}
\resizebox{\textwidth}{!}{\begin{tabular}{lcccccc}
\hline
\multicolumn{1}{c}{\multirow{2}{*}{Operator}}               & \multicolumn{2}{c}{$K_u = K_d = 3$}                                          & \multicolumn{2}{c}{$K_u = K_d = 5$}                                         & \multicolumn{2}{c}{$K_u = K_d = 10$}                                        \\ \cline{2-7} 
\multicolumn{1}{c}{}                                        & \multicolumn{1}{c}{$M = 30 , N = 8$} & \multicolumn{1}{c}{$M = 60, N = 40$} & \multicolumn{1}{c}{$M = 30, N = 8$} & \multicolumn{1}{c}{$M = 60, N = 40$} & \multicolumn{1}{c}{$M = 30, N = 8$} & \multicolumn{1}{c}{$M - 60, N = 40$} \\ \hline
DE\textbackslash{}rand\textbackslash{}1                                                           &  $22.23\pm3.36$                                     & $22.79\pm5.10$                        &       $28.07\pm3.69$                               & $29.45\pm2.68$                        &       $38.64\pm3.14$                               & $49.96\pm6.50$                        \\
DE\textbackslash{}rand\textbackslash{}2                       &   $22.09\pm3.32$                                    & $22.82\pm5.00$                        &       $27.71\pm3.40$                               & $29.54\pm3.12$                        &    $37.29\pm2.19$                                  & $49.24\pm7.55$                        \\
DE\textbackslash{}best\textbackslash{}1                       & $21.96\pm3.41$                                      & $22.53\pm5.22$                        &         $27.30\pm3.27$                             & $29.01\pm2.71$                        &        $37.23\pm1.80$                              & $46.99\pm6.37$                        \\
DE\textbackslash{}best\textbackslash{}2                       & $22.05\pm3.27$                                      & $23.08\pm5.42$                        &         $28.15\pm3.93$                             & $28.45\pm2.74$                        &       $38.90\pm2.81$                               & $49.38\pm6.30$                        \\ 
\multicolumn{1}{l}{CHDE} & \multicolumn{1}{c}{$\mathbf{23.52\pm4.44}$}                 & \multicolumn{1}{c}{$\mathbf{24.85\pm6.52}$}   & \multicolumn{1}{c}{$\mathbf{30.71\pm4.43}$}                & \multicolumn{1}{c}{$\mathbf{31.93\pm3.56}$}   & \multicolumn{1}{c}{$\mathbf{42.59\pm3.61}$}                & \multicolumn{1}{c}{$\mathbf{54.39\pm7.98}$}   \\ \hline
\end{tabular}}
\vspace{ 0.5em}
\label{tab:DE_operator}
\end{table*}

\begin{figure*}[t]
\vspace{-0.2cm}
\centering

\begin{minipage}[t]{0.66\textwidth} 
    \centering
    \begin{subfigure}[t]{0.49\linewidth}
        \centering
        \includegraphics[trim=0 0cm 0cm 0cm,clip,width=\linewidth]{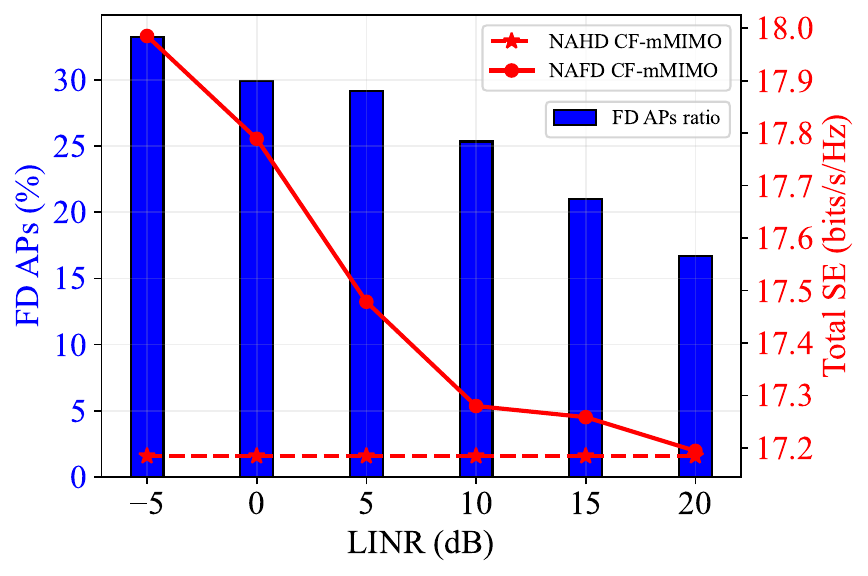}
        \vspace{-0.5em}
        \caption{$M = 20, N = 5$}
        \label{fig:time_slot}
    \end{subfigure}\hfill
    \begin{subfigure}[t]{0.49\linewidth}
        \centering
        \includegraphics[trim=0 0cm 0cm 0cm,clip,width=\linewidth]{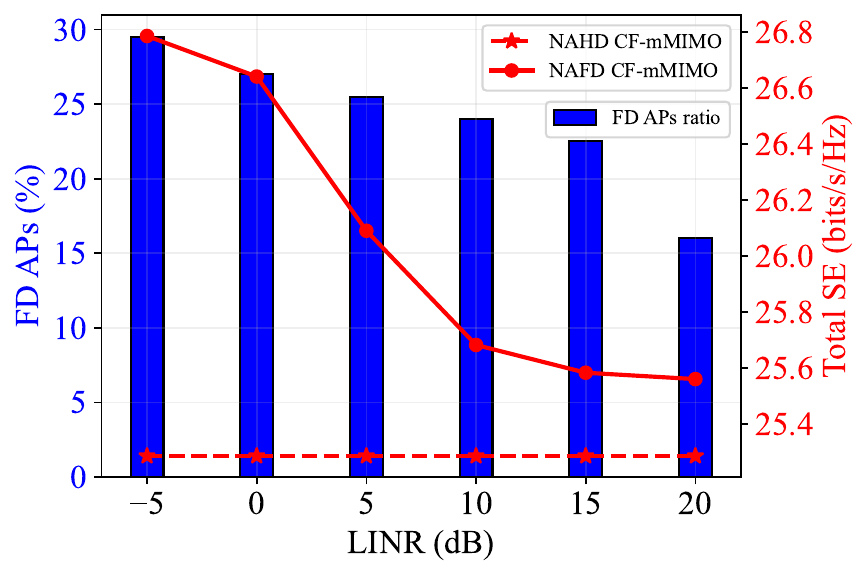}
        \vspace{-0.5em}
        \caption{$M = 40, N = 5$}
        \label{fig:large-scale}
    \end{subfigure}

    \vspace{-0.6em}
\caption{Efficiency of the proposed model compared to the previous model [19] that only includes APs operating in HD mode.}
\label{fig:FD_model}
\end{minipage}
\hfill
\begin{minipage}[t]{0.32\textwidth} 
    \centering
    \includegraphics[trim=0 0cm 0cm 0cm,clip,width=0.9\linewidth]{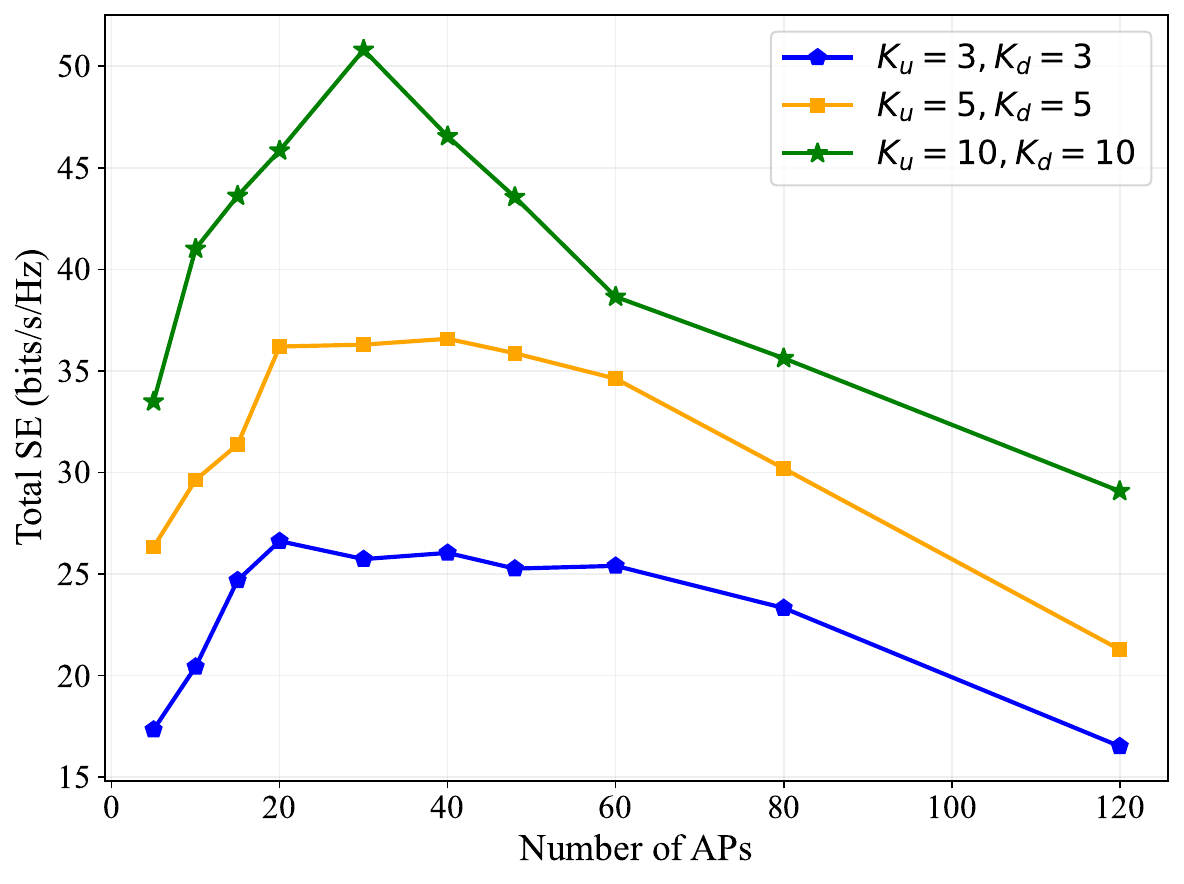}
    \vspace{0.5em}
    \captionof{figure}{Impact of number of APs ($MN = 240$).}
    \label{fig:impact_of_number_of_APs}
\end{minipage}

\vspace{-1.0em}
\end{figure*}

Figure~\ref{fig:objective_otpimization} presents the total SE of all UEs under various system configurations characterized by the number of APs and DL/UL UEs. This paper considers that there are $240$ antennas $(MN = 240)$ divided equally for each AP. The proposed CHDE consistently outperforms the benchmark schemes, achieving performance gains of up to $28\%$, $29\%$, and $82\%$ over CHPSO, CHGA, and Random-NAFD, respectively. For the heuristic scheme Random-NAFD, the average total SE is lower because of the heterogeneous distribution of UEs relative to the APs and the equal power allocation applied without LSFD. Therefore, this verified the necessity for joint optimization of power control and LSFD weight design. Compared with the robust benchmarks, CHDE yields $10\%$–$30\%$ improvements. Problem \eqref{P:SE} is a combination optimization. However, the problem can be regarded as an approximately continuous real optimization when expressed through the individual’s representation. Hence, CHDE outperforms CHPSO and CHGA, as DE is well-suited for continuous variables.

In Fig.~\ref{fig:Convergence}, we plot the convergence behavior of CHGA, CHPSO, and the proposed CHDE, all initialized with identical parameters. While each algorithm significantly improves the total SE throughout the evolutionary process, CHDE surpasses the others in the quality of the final solution. In the early generations, CHGA and CHDE converge faster than CHPSO; however, CHGA exhibits premature convergence as the algorithm becomes trapped in a local optimum. CHPSO improves monotonically over generations, but the gains remain marginal. In contrast, CHDE demonstrates superior ability to escape local optima and enters the stable convergence phase around the $100$-th generation. While CHGA shows faster progress initially, premature convergence limits its effectiveness. In contrast, CHDE ultimately achieves superior performance in both convergence stability and final solution quality.

Figure~\ref{fig:run_time} illustrates the average runtime of CHDE under various scenarios. It is observed that the runtime is influenced by the number of UEs, the number of APs, and the number of antennas per AP. The algorithm runs efficiently when the number of UEs is relatively small, such as $6$ or $10$ UEs, and when the number of APs is appropriately configured. However, as the number of users increases, e.g., to $20$ UEs, the resource allocation process becomes more complex, and evaluating the fitness of individuals consumes more computational resources, leading to a significant increase in runtime. Moreover, increasing the number of APs also increases the algorithm's complexity, while the system performance tends to degrade. For instance, with $6$ UEs, the runtime of CHDE is $52.45$s with $30$ APs, but the runtime increases by $12$ times when using $120$ APs, while the total SE decreases by approximately $1.47$ times. Similarly, for a larger-scale scenario with $20$ UEs, the runtime is $289.2$s with $30$ APs, whereas increasing the number of APs to $120$ results in a $41.5$-fold increase in runtime due to the significantly larger number of optimization variables. Therefore, it is important to configure an appropriate number of APs for the system. Moreover, Fig.~\ref{fig:Convergence} indicates that CHDE enters a stable convergence regime after roughly the $100$-th generation, beyond which performance gains become marginal. This suggests that early-stopping criteria can be employed in practice to reduce the runtime, e.g., when the fitness improvement between successive generations becomes negligible, while preserving near-optimal performance, thereby enhancing the practicality of CHDE for large-scale network optimization.

\begin{figure*}
\centering
\begin{subfigure}{0.31\textwidth}
    \includegraphics[width=\linewidth]{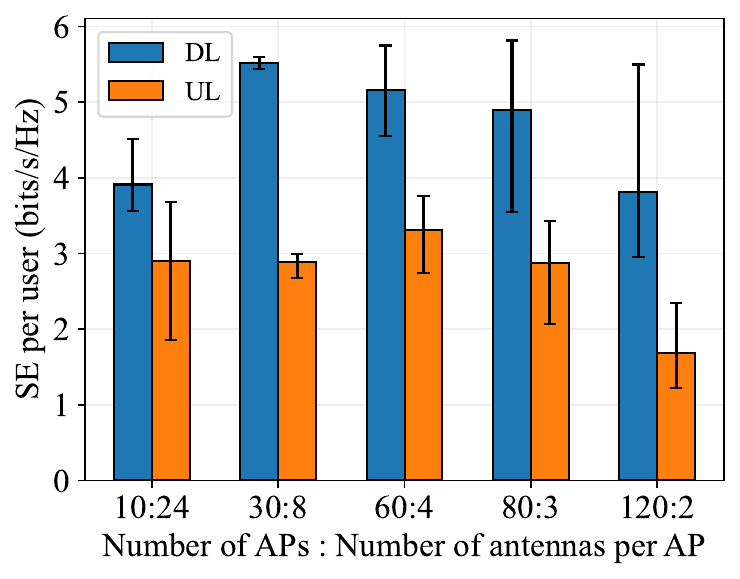}
    \caption{$K_u = 3, K_d = 3$}
    \label{fig:first_SE_per_UE}
\end{subfigure}
\begin{subfigure}{0.31\textwidth}
    \includegraphics[width=\textwidth]{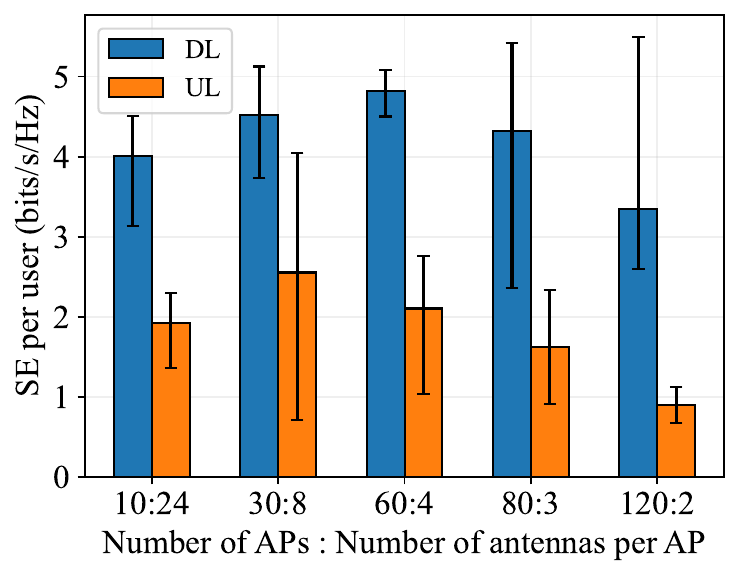}
    \caption{$K_u = 5, K_d = 5$}
    \label{fig:second_SE_per_UE}
\end{subfigure}
\begin{subfigure}{0.31\textwidth}
    \includegraphics[width=\textwidth]{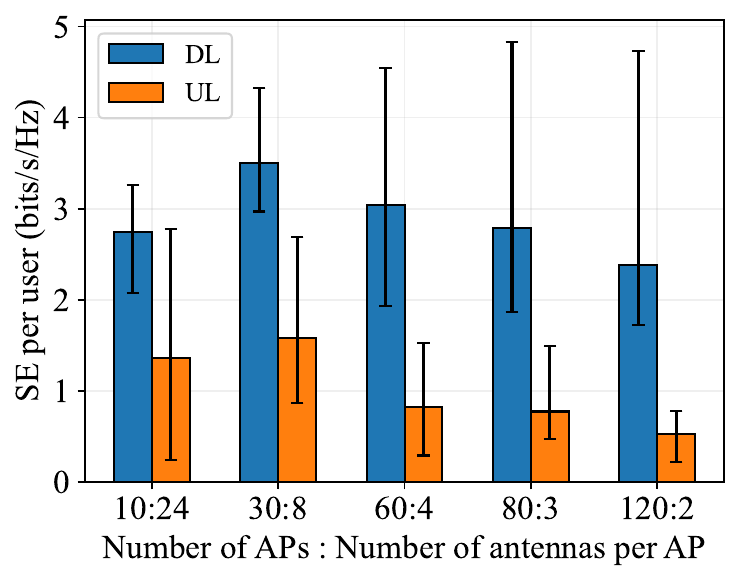}
    \caption{$K_u = 10, K_d = 10$}
    \label{fig:third_SE_per_UE}
\end{subfigure}    
\caption{Average SE per UE.}
\label{fig:SE_per_UE}
\vspace{-1em}
\end{figure*}

Table~\ref{tab:DE_operator} summarizes the performance of various DE operators used for reproduction. The $rand/1$ and $rand/2$ strategies exhibit strong exploration abilities. However, as the level of randomness increases in $rand/2$, the SE decreases, indicating that excessive randomness can impair the ability of the algorithm to exploit promising solution regions, thereby slowing convergence. In contrast, the $best/1$ operator, which relies entirely on the best individual, tends to converge quickly but often suffers from premature convergence, limiting the overall effectiveness. It is observed that introducing a random factor in $best/2$ enhances the SE by enabling solutions to escape the local region. To leverage different strategies, CHDE enhances diversity by selecting reproduction candidates from among the top individuals, rather than relying solely on the best one. In addition, a random component is introduced to enhance the exploration capability. Experimental results validate that CHDE consistently outperforms the compared operators by approximately $5.8\%$–$14.39\%$.

 Figure~\ref{fig:FD_model} compares the proposed model with the previous model in \cite{Mohammad:JSAC:2023}, where only HD-mode APs are employed. The proportion of APs operating in FD mode ranges from $16\%$ to $33\%$, depending on system configuration, and this ratio increases with higher AP density. The performance gain of the proposed model is affected by both the loop interference-to-noise ratio (LINR) and the number of APs. When LINR is higher than noise, system performance tends to decrease, resulting in lower SE. In contrast, increasing the number of FD APs generally improves the relative efficiency of the proposed model. 

Figure~\ref{fig:impact_of_number_of_APs} illustrates the impact of the number of APs on the total SE of the network system. It shows that for each scenario ($K_u = K_d \in [3, 5, 10]$), there is an optimal number of antennas per AP that maximizes the total SE. Initially, when the number of APs increases, the total SE also increases across all scenarios. However, after reaching a specific number of APs, the total SE no longer increases and may even decrease when this amount becomes too large since the total number of antennas across the APs remains constant ($MN = 240$). The presence of larger residual SI and CLI negatively impacts the total SE. Additionally, when individual SE lower constraints are imposed, the system needs to allocate more power to UEs with unfavorable links. Then, stronger interference exists, and the total SE is sacrificed to compensate for these UEs.

\begin{figure*}
\centering
\begin{subfigure}{0.31\textwidth}
    \includegraphics[width=\linewidth]{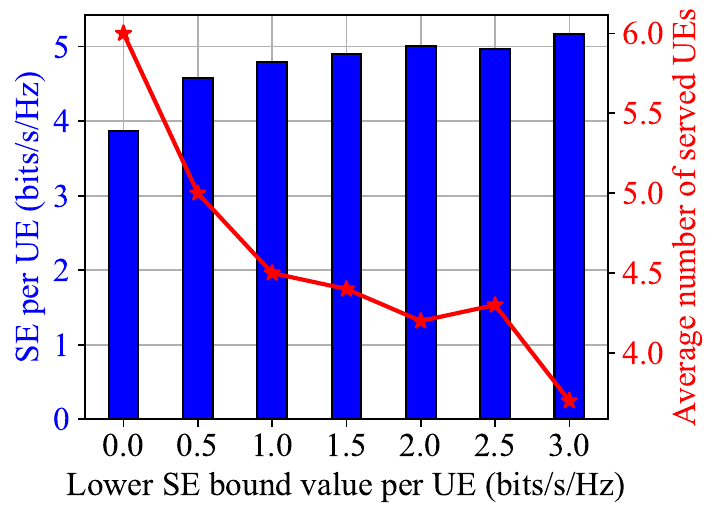}
    \caption{$K_u = 3, K_d = 3$}
    \label{fig:first_constraint}
\end{subfigure}
\begin{subfigure}{0.31\textwidth}
    \includegraphics[width=\textwidth]{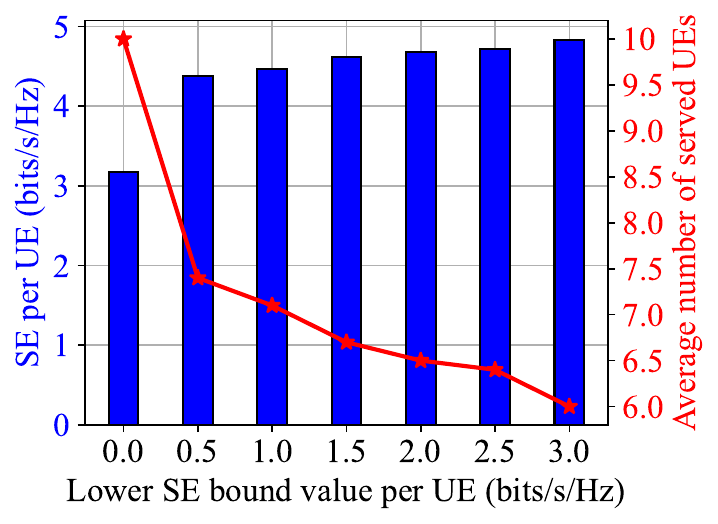}
    \caption{$K_u = 5, K_d = 5$}
    \label{fig:second_constraint}
\end{subfigure}
\begin{subfigure}{0.31\textwidth}
    \includegraphics[width=\textwidth]{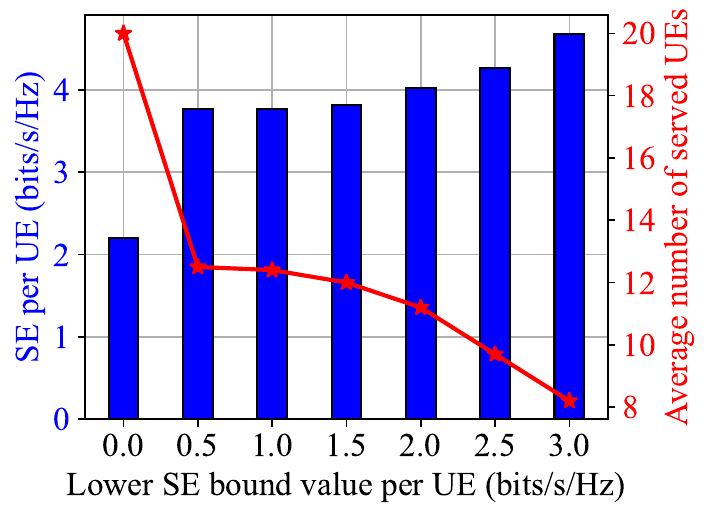}
    \caption{$K_u = 10, K_d = 10$}
    \label{fig:third_constraint}
\end{subfigure}
  \vspace{0em}      
\caption{Impact of SE lower requirement on the average of SE per UE and the number of served UEs.}
\label{fig:Constraint}
\vspace{-0.2em}
\end{figure*}

\begin{figure*}[t]
\vspace{-0.2cm}
    \centering
    \begin{minipage}[t]{0.32\textwidth}
        \centering
        \includegraphics[trim=0 0cm 0cm 0cm,clip,width=1\textwidth]{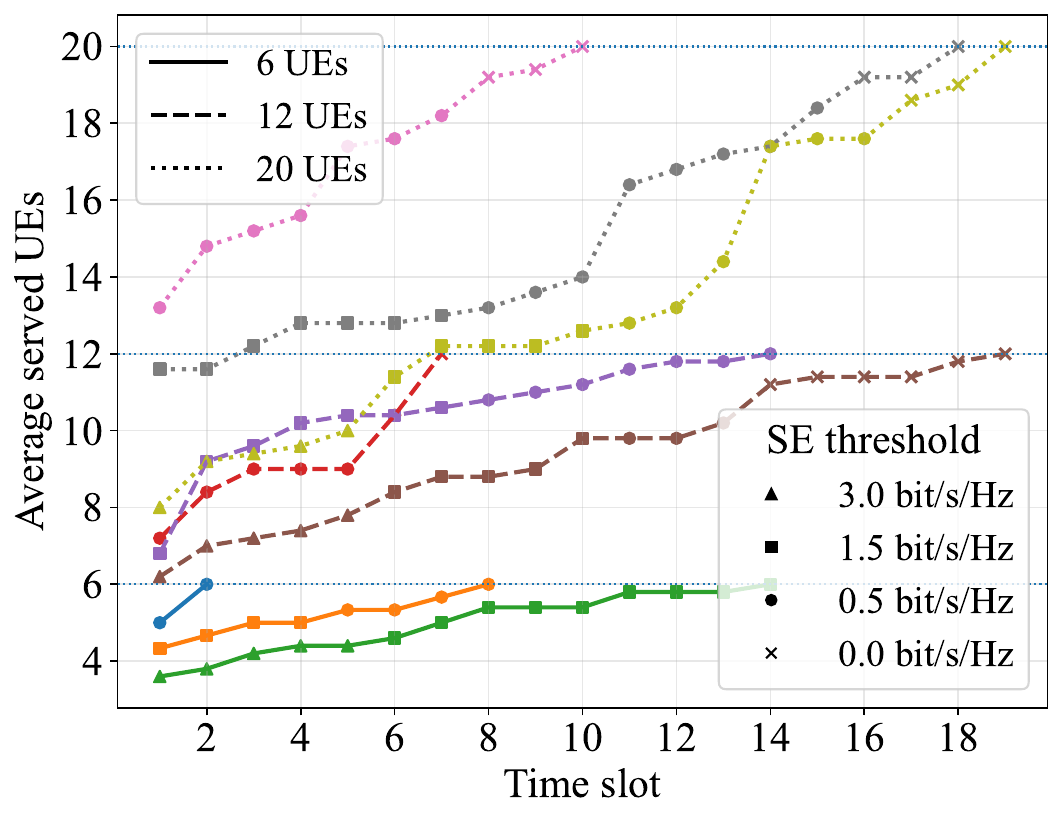}
        \vspace{-1.5em}
	    \caption{Average served UEs versus number of time slots for different network sizes and initial SE requirements.}
    \label{fig:time_slot}
    \end{minipage}
    \hfill
    \begin{minipage}[t]{0.32\textwidth}
        \centering
        \includegraphics[trim=0 0cm 0cm 0cm,clip,width=1.05\textwidth]{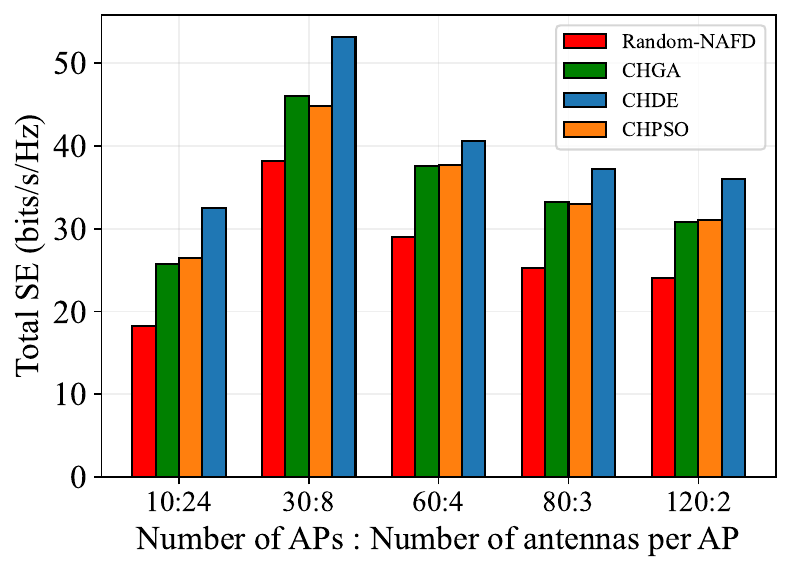}
        \vspace{-1.5em}
    \caption{Large-scale system with $100$ users.}
    \label{fig:large-scale}
    \end{minipage}
    \hfill
    \begin{minipage}[t]{0.32\textwidth}
        \centering
        \includegraphics[trim=0 0cm 0cm 0cm,clip,width=1\textwidth]{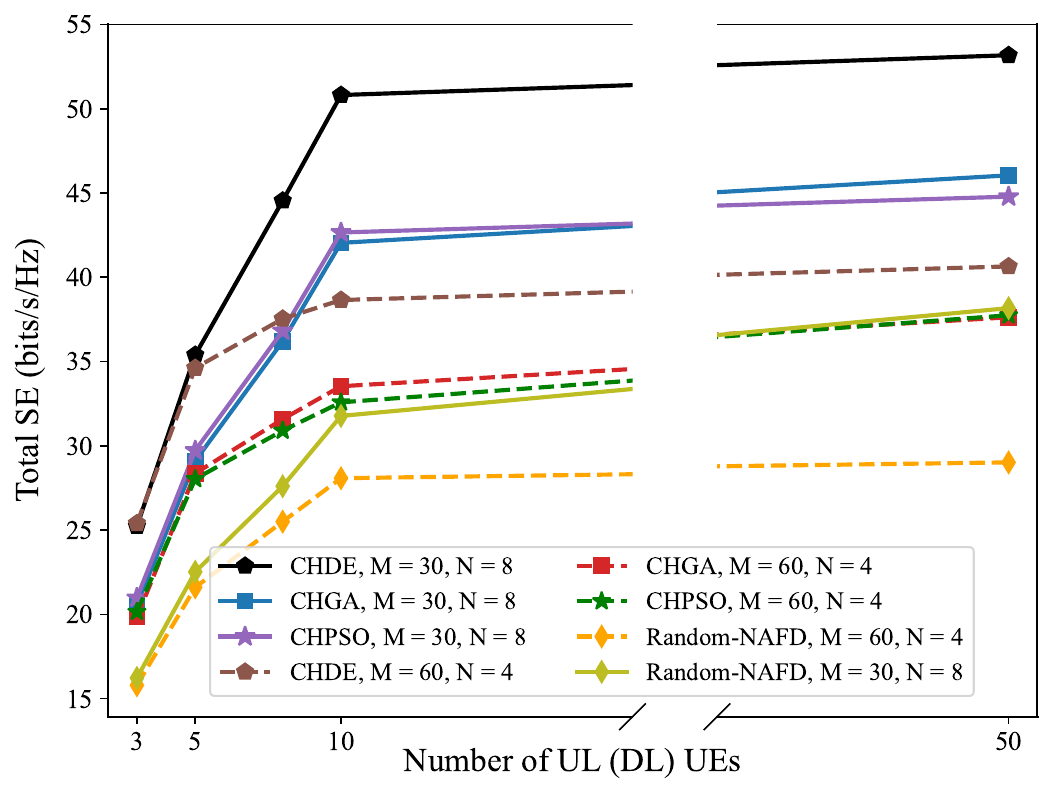}
        \vspace{-1.5em}
	\caption{Impact of the number of UEs.}
    \label{fig:impact_of_number_of_UEs}
    \end{minipage}
\vspace{-1.0em}
\end{figure*}

Figure~\ref{fig:SE_per_UE} demonstrates the SE per UE with lower threshold $\mathcal{S}_{QoS} = 0.2$~bits/s/Hz. With the same power provided to the system, when the number of UEs increases, SE per UE slightly decreases. Additionally, SE per UE is also influenced by the AP-antenna configuration implemented in the network, which aligns with the analysis of total SE as indicated in Fig.~\ref{fig:objective_otpimization}. With $M = 30$ and $N = 8$, the SEs per UE in three scenarios with 6, 10, and 20 UEs are $4.21$, $3.54$, and $2.54$~bits/s/Hz, respectively. There is a significant disparity between SE of DL UEs and SE of UL UEs. Moreover, as the number of UEs and APs increases, the gap between the maximum SE and minimum SE per UE becomes more pronounced. With 20 UEs ($K_d = K_u = 10$), the max-min distance ranges from $0.56~\text{to}~3$~bits/s/Hz. High UE density, combined with a small number of antennas per AP, presents a significant challenge for the network system in allocating power to reduce interference and ensure specific SE requirements for UEs.

CHDE suspends service to UEs whose SE requirements cannot be met, thereby improving network quality and enhancing resource efficiency. Figure~\ref{fig:Constraint} illustrates the impact of the minimum SE constraints on average SE and the number of UEs served with $\mathcal{S}_{QoS} \in [0, 3]$~bits/s/Hz. Most UEs are served when the constraints are relaxed. However, as the minimum SE threshold increases, it becomes more challenging to serve the UEs. This effect is even more evident in networks with high UE density. Figure~\ref{fig:time_slot} illustrates the average number of time slots required to serve all UEs under different initial SE thresholds. In this experiment, UEs that are not scheduled in a given slot are retained with new UEs for a subsequent scheduling slot. We evaluate three network sizes with initial SE thresholds of 0.5, 1.5, and 3.0 bit/s/Hz. To prevent long-term starvation, a re-entry mechanism with threshold relaxation is introduced. If a UE remains unserved for several consecutive time slots due to unfavorable channel conditions, its SE requirement is gradually relaxed to a lower level, ensuring eventual service. For instance, in the 6-UE scenario, all users are served after an average of 2 time slots with a 0.5 bit/s/Hz threshold, whereas approximately 14 time slots are required with a 3.0 bit/s/Hz threshold. In the latter case, the SE constraint of persistently unserved UEs is relaxed to 1.5 bit/s/Hz after several slots. Similar behavior is observed in the 12-UE and 20-UE cases. These results demonstrate that the proposed relaxation mechanism guarantees that no UE is permanently excluded.

Additionally, this paper evaluates the performance of the proposed algorithm in a large-scale system scenario, where the system serves $100$ UEs ($K_u = K_d = 50$) with a LINR of $20$ dB, as depicted Fig.~\ref{fig:large-scale}. The results demonstrate the superior effectiveness of the proposed algorithm compared to other algorithms. The total SE reaches its maximum when the system includes $30$ APs. However, the improvement of the total SE is not significantly higher than in the scenario with $20$ UEs, due to the same power limitation of the system. The impact of the number of UEs on the system performance is depicted in Fig.~\ref{fig:impact_of_number_of_UEs}. The total SE of the system is improved with a higher number of UEs. This figure also shows the gain of CHDE compared to other benchmark schemes with two AP configurations ($M = 30, 60$) in various UE density scenarios. The balance between the number of APs and the number of antennas per AP influences the rate of increase in total SE. However, the total SE increases rapidly as the number of UEs grows from $6$ to $20$, but the rate of improvement slows down when the system serves up to $100$ UEs. The total SE tends to saturate as the number of UEs exceeds a certain threshold, due to power limitations and increased interference.

\vspace{-0.3em}
\section{Conclusion}~\label{Sec:conc}
This paper analyzed and optimized the ergodic SE in NAFD CF-mMIMO systems, which can flexibly operate in FD or HD mode with linear processing. By jointly designing the AP mode of operation, power control, and LSFD weights, this research maximizes the total SE, while accounting for individual UL and DL QoS SE requirements as well as UL and DL transmit power constraints. We derived the analytical SE for each UE in this proposed model under arbitrary fading channel models and the impact of loop interference caused by FD APs. Leveraging the robustness of evolutionary computation, CHDE was proposed to maximize the total SE while satisfying specific constraints. The algorithm employs a repair mechanism to handle constraint violations, which may exclude some UEs due to practical limits. Numerical results verify the model’s superiority over prior systems and demonstrate CHDE’s efficiency against benchmarks.

\vspace{-0.3em}
\appendices

\section{Proof of Proposition~\ref{Prop:SE:DLPZF}}~\label{Proof:Prop:SE:DLPZF}
By applying the use-and-then-forget capacity bounding method from~\cite{ngo17TWC}, the associated SE of $\UEdk$ is expressed as 
\vspace{0em}
\begin{align}~\label{eq:Sdlk1:PZF}
\mathcal{S}_{\dl,k}  =  \frac{\tau_c-\tau_t}{\tau_c}
\log_2 \bigg(1 +  \frac{\vert \mathbb{DS}_{\dl,k}^\PZF \vert^2}
{ \Delta_{\dl}}\bigg),
\vspace{1em}
\end{align}
where $\Delta_{\dl}=\Ex\big\{\!\vert \mathbb{BU}_{\dl,k}^\PZF \vert^2 \!\big\}
\!\!+\!\!\!\!\sum\nolimits_{k'\in\mathcal{K}_d \setminus k}\Ex\big\{\!\vert \mathbb{DI}_{\dl,kk'}^\PZF \!\vert^2\big\}   
\!\!+\!\!\!\! \sum\nolimits_{\ell\in \mathcal{K}_{u}}\Ex\big\{\!\vert \mathbb{UI}_{\ell}^\PZF \vert^2\!\big\}
+\!1$.
Therefore, we need to compute $\mathbb{DS}_{\dl,k}$, $\Ex\big\{ \vert \mathbb{BU}_{\dl,k}\vert^2\big\}$, $\Ex\big\{\vert\mathbb{DI}_{\dl,kk'}\vert^2\big\}$, and $\Ex\big\{\vert \mathbb{UI}_{\ell} \vert^2\big\}$. To compute $\mathbb{DS}_{\dl,k}$, by replacing $\gmkd=(\hgmkd+\tgmkd)$, we get
\begin{align}~\label{eq:DSk:PZF}
  \mathbb{DS}_{\dl,k} 
&\!=\! \sqrt{\rho_d}\bigg(\!\sum\limits_{m \in \Zk}\!\!\! a_m\theta_{mk} \gamdmk
\!+\!\!
N\!\!\!\!\sum\limits_{m \in \Tk} \!\!\!\! a_m\theta_{mk} \gamdmk\!\bigg),
\end{align}
which is based on the observation that $\wmkdlzf$, $\wmkdlmr$, and $\tgmkd$ are zero-mean  RVs and are independent of each other.

Note that $\Ex\left\{ \vert \mathbb{BU}_{\dl,k}\vert^2\right\} =
{\rho_d} \Xi_k - \vert\mathbb{DS}_{\dl,k}\vert^2$, where
\vspace{-0.1em}
\begin{align}~\label{eq:Xik:PZF}
\Xi_k 
&=\Xi_k^\ZF + \Xi_k^\MR  + 2\Big(\sum\nolimits_{m \in \Zk} a_m\theta_{mk} \gamdmk\Big)
\nonumber\\
&\hspace{4em}\times
\Big(N\sum\nolimits_{m \in \Tk}  a_m\theta_{mk} \gamdmk\Big),
\end{align}
with $\Xi_k^\ZF\triangleq\Ex\big\{
\big\vert\sum_{m \in \Zk} a_m\theta_{mk}
\left(\gmkd\right)^\dag\wmkdlzf\big\vert^2\big\}$ and
$\Xi_k^\MR\triangleq\Ex\big\{ \big\vert\sum_{m \in \Tk} a_m\theta_{mk}
\left(\gmkd\right)^\dag\wmkdlmr\big\vert^2\big\}$. Additionally,
\vspace{-0.2em}
\begin{align}~\label{eq:XikZF}
    \Xi_k^\ZF&=\Ex\Big\{
\Big\vert\sum\nolimits_{m \in \Zk} a_m\theta_{mk}
\left(\hgmkd+\tgmkd\right)^\dag\wmkdlzf\Big\vert^2\Big\}
\nonumber\\
&=\Big(\sum\nolimits_{m \in \Zk} a_m\theta_{mk}\gamdmk\Big)^2 
\nonumber\\
&
+ \sum\nolimits_{m \in \Zk} a_m\theta_{mk}^2 (\betamkd-\gamdmk)\frac{\gamdmk}{N - \vert \Sm\vert},
\end{align}
where we have used~\eqref{eq:PZF_prec2} and
\begin{align}~\label{eq:gtildevzf}
&\sum\nolimits_{m \in \Zk} a_m\theta_{mk}^2
\Ex\Big\{ \Big\vert
\left(\tgmkd\right)^\dag\wmkdlzf\Big\vert^2\Big\}
\nonumber\\
&=\!\!\!
\sum_{m \in \Zk} \!\!\!a_m\theta_{mk}^2\!
\left(\betamkd \!-\! \gamdmk\right)
\!(\gamdmk)^2
\Ex\Big\{\!\Big[\!\Big( \!\big(\hGmdl\big)^\dag \hGmdl\!\Big)^{\!\!-1}\Big]_{kk}\Big\},
\nonumber\\
&=
\sum\nolimits_{m \in \mathcal{M}} a_m\theta_{mk}^2
\left(\betamkd - \gamdmk\right)
\frac{\gamdmk}{N-\vert\Sm\vert},
\end{align}
where, $\tgmkd$ and $\wmkdlzf$ are independent, and the result follows from~\cite[Lemma 2.10]{tulino04}. Moreover, $ \Xi_k^\MR$ is derived as
\begin{align}~\label{eq:XikMR}
 \Xi_k^\MR
&
=\text{var}\Big(\sum\nolimits_{m \in \Tk} a_m\theta_{mk}
\left(\hgmkd+\tgmkd\right)^\dag\wmkdlmr\Big)
\nonumber\\
&
+
\bigg\vert\Ex\bigg\{
\sum\nolimits_{m \in \Tk} a_m\theta_{mk}
\left(\hgmkd+\tgmkd\right)^\dag\wmkdlmr\bigg\}\bigg\vert^2
\nonumber\\
&=\sum\nolimits_{m \in \Tk}\!\!a_m\theta_{mk}
\Big( 
\Ex\Big\{\!\big\Vert\hgmkd\big\Vert^4\!\Big\}\!+\!\Ex\Big\{\big\vert(\tgmkd)^\dag\hgmkd\!\big\vert^2\!\Big\} 
\nonumber\\
&
-\! (N\gamdmk)^{\!2}\!\Big)
\!+\!
\bigg\vert\Ex\bigg\{\!\!
\sum_{m \in \Tk} \!\!\!a_m\theta_{mk}
\!\left(\hgmkd\!+\!\tgmkd\right)^{\!\dag}\wmkdlmr\!\bigg\}\!\bigg\vert^2
\nonumber\\
&=\!\!\!\sum\limits_{m \in \Tk}\!\!\!a_m\theta_{mk}^2
N\gamdmk\betamkd
\!+\!\Big(\!N\!\!\!\sum\limits_{m \in \Tk} \!\!\!a_m\theta_{mk}\gamdmk\Big)^{\!2}\!. 
\end{align}
Therefore, plugging~\eqref{eq:XikZF} and~\eqref{eq:XikMR} into~\eqref{eq:Xik:PZF}, we have
\begin{align}~\label{eq:EBUk:PZF:final}
\Ex\left\{ \vert \mathbb{BU}_{\dl,k}\vert^2\right\} &=
{\rho_d} \Big( \sum\nolimits_{m \in \Zk} a_m\theta_{mk}^2 (\betamkd-\gamdmk)\nonumber\\
&\hspace{-4em}
\times\frac{\gamdmk}{N - \vert \Sm\vert} 
+
\sum\nolimits_{m \in \Tk}a_m\theta_{mk}^2
N\gamdmk\betamkd \Big). 
\end{align}
Similarly, we can compute $\Ex\left\{\vert \mathbb{UI}_{\ell} \vert^2\right\} ={\rho_u {\varsigma}_\ell}\betakldu$  and 
\begin{align}~\label{eq:EDIkkp:PZF}
\Ex\left\{ \vert \mathbb{DI}_{\dl,kk'} \vert^2\right\} &\!=\!\!
{\rho_d} \Big( \!\sum\nolimits_{m \in \Zk} \!\!a_m\theta_{mk'}^2 (\betamkd\!-\!\gamdmk)\frac{\gamdmkp}{N \!-\! \vert \Sm\vert}
\nonumber\\
&
+
N\sum\nolimits_{m \in \Tk}a_m\theta_{mk'}^2
\gamdmk\betamkpd \Big).
\end{align}
Using~\eqref{eq:DSk:PZF}, \eqref{eq:EBUk:PZF:final}, and~\eqref{eq:EDIkkp:PZF} into~\eqref{eq:Sdlk1:PZF}, the desired result is in~\eqref{eq:DL:SE}.

\section{Proof of Theorem~\ref{theo:first_fitting_time}}
\label{proof:first_hitting_time}
From the average gain, an upper bound on $\Ex[T_0]$ is obtained
\begin{Lemma}\label{lemma:first_hitting_time_proof}
    \cite{yushan2016first} Let $\{\zeta
    ^{(G)}\}_{G = 1}^{\infty}$ be a stochastic non-negative process. Let $T_0^{\zeta} = \min\{G | \zeta^{(G)} = 0\}$ and $\sigma-$algebra $\mathcal{H}^{(G)} = \sigma(\zeta^{(1)}, \zeta^{(2)}, \ldots, \zeta^{(G)})$. Assume that $\mathbb{E}\{T^{\zeta}_{0}\} < +\infty$. If there exists $\alpha \in \mathbb{R}$ and $\mathbb{E}\{\zeta^{(G)} - \zeta^{(G+1)} |\mathcal{H}^{(G)}\} \geq \alpha > 0, \forall G \geq 1$, then $\mathbb{E}\{T_0^{\zeta} |\zeta_1\} \leq \frac{\zeta_1}{\alpha}$.
\end{Lemma}
Let $\chi(x) = 0$ if $x=0$. Otherwise, $\chi(x) = 1 + \int_{0}^{x}\frac{1}{h(t)}dt$ if $x>0$. We consider the two cases: \textit{(i)} If $\Delta^{(G)}_{\varepsilon} > 0$ and $\Delta^{(G+1)}_{\varepsilon} = 0$,
    $\mathbb{E}\{\chi(\Delta^{(G)}_{\varepsilon}) - \chi(\Delta^{(G+1)}_{\varepsilon}) |\mathcal{F}^{(G)} \} =\mathbb{E} \{ 1 + \int_{0}^{\Delta^{(G)}_{\varepsilon}} \frac{1}{h(t)}dt|\mathcal{F}^{(G)} \} \geq 1$,
where $\mathcal{F}^{(G)} = \sigma(\Delta^{(1)}_{\varepsilon}, \Delta^{(2)}_{\varepsilon}, \ldots,\Delta^{(G)}_{\varepsilon})$; and \textit{(ii)} If $\Delta^{(G)}_{\varepsilon} > 0$ and $\Delta^{(G+1)}_{\varepsilon} > 0$, it holds that
$\mathbb{E} \{\chi(\Delta^{(G)}_{\varepsilon}) \!-\! \chi(\Delta^{(G+1)}_{\varepsilon}) |\mathcal{F}^{(G)} \}
    =\mathbb{E} \{ \int_{\Delta^{(G+1)}_{\varepsilon}}^{\Delta^{(G)}_{\varepsilon}}\frac{1}{h(t)}dt \big|\mathcal{F}^{(G)} \} 
    \stackrel{(a)}{\geq} \mathbb{E} \{\frac{\Delta^{(G)}_{\varepsilon} - \Delta^{(G+1)}_{\varepsilon}}{h(\Delta^{(G)}_{\varepsilon})} \big| \mathcal{F}^{(G)} \} 
    = \frac{1}{h(\Delta^{(G)}_{\varepsilon})} \mathbb{E} \{\Delta^{(G)}_{\varepsilon} \!-\! \Delta^{(G+1)}_{\varepsilon} \big| \mathcal{F}^{(G)} \} \geq 1,$
where $(a)$ is obtained by $h(\Delta^{(G)}_{\varepsilon}) > h(t), \forall t \in [\Delta^{(G+ 1)}_{\varepsilon}, \Delta^{(G)}_{\varepsilon}]$.

From \textit{(i)} and \textit{(ii)}, $\mathbb{E}\left\{\chi(\Delta^{(G)}_{\varepsilon}) - \chi(\Delta^{(G+1)}_{\varepsilon}) |\mathcal{F}^{(G)}\right\} \geq 1, \forall \Delta^{(G)}_{\varepsilon} > 0$. Considering a stochastic process $\left\{\chi(\Delta^{(G)}_{\varepsilon})\right\}_{G = 1}^{\infty}$, denote $T_{0}^{\chi} = \min \left\{G \big| \chi(\Delta^{(G)}_{\varepsilon}) = 0\right\}$ as the first hitting time. Using Lemma \ref{lemma:first_hitting_time_proof} and $\mathbb{E}\left\{\DElta^{(G)}_{\varepsilon} - \Delta^{(G + 1)}_{\varepsilon}\big| \mathcal{F}^{(G)} \right\} = \mathbb{E}\left\{\Delta^{(G)}_{\varepsilon} - \Delta^{(G + 1)}_{\varepsilon}\big| \Delta^{(G)}_{\varepsilon} \right\}$, Theorem \ref{theo:first_fitting_time} is proven
\begin{align}
    \mathbb{E}\left\{T_\varepsilon \big| \Delta^{(1)}_{\varepsilon} \right\} 
    &= \mathbb{E}\left\{T_{\varepsilon}^{\chi} \big| \chi(\Delta^{(1)}_{\varepsilon})\right\} \nonumber
    &
    \leq 1 + \int_{0}^{\Delta^{(1)}_{\varepsilon}}\frac{1}{h(t)}dt.
\end{align}
\bibliographystyle{IEEEtran}
\bibliography{IEEEabrv,references}

\end{document}